\documentclass[11pt]{article}
\usepackage{fullpage}

\usepackage{subfig}
\usepackage[usenames,dvipsnames]{xcolor}
\usepackage[colorlinks,citecolor=blue,linkcolor=BrickRed]{hyperref}
	\usepackage{latexsym,graphicx,epsfig,color}
\usepackage{amsfonts,amssymb,amsmath,amsthm,amstext}
\usepackage{libertine}
\usepackage[libertine]{newtxmath}
\usepackage{enumitem}
\setitemize{itemsep=2pt,topsep=0pt,parsep=0pt}
\usepackage{url,setspace}
\usepackage{multirow}
\usepackage{rotating}
\usepackage{makeidx}
\usepackage{accents}
\usepackage{xspace}
\usepackage{algorithm,algpseudocode}
\usepackage{bm}
\usepackage{changepage} 	
\usepackage{thmtools,thm-restate}
\usepackage{nicefrac}
\usepackage{cleveref}
\usepackage{import}

\newcommand{\IGNORE}[1]{}
\allowdisplaybreaks

\makeindex


\newtheorem{theorem}{Theorem}[section]
\newtheorem{claim}[theorem]{Claim}

\newtheorem{lemma}[theorem]{Lemma}

\newtheorem{fact}[theorem]{Fact}
\theoremstyle{definition}

\newtheorem{defn}[theorem]{Definition}

\usepackage{thmtools,thm-restate} 

\newcommand{\E}{\mathbb{E}}

\newcommand{\calD}{\mathcal{D}}

\newcommand{\calM}{\mathcal{M}}

\def\e {\varepsilon}
\def\eps {\epsilon}

\newcommand{\poly}{\operatorname{poly}}

\newcommand{\loglog}{\log\log}

\usepackage{wrapfig}
\usepackage[compact]{titlesec}

\newcounter{note}[section]

\newcommand{\initOneLiners}{%
    \setlength{\itemsep}{0pt}
    \setlength{\parsep }{0pt}
    \setlength{\topsep }{0pt}
}
\newenvironment{OneLiners}[1][\ensuremath{\bullet}]
    {\begin{list}
        {#1}
        {\initOneLiners}}
    {\end{list}}

\title{Robust Algorithms for the Secretary Problem}

	\author{Domagoj Bradac\thanks{
	(domagoj.bradac@gmail.com)
        Department of Mathematics, Faculty of Science, University of
        Zagreb. Part of this work was done when visiting the Computer
        Science Department at Carnegie Mellon University.
	}
	\and Anupam Gupta\thanks{
        (anupamg@cmu.edu)
        Computer Science Department,
        Carnegie Mellon University. Supported in part by NSF award CCF-1907820.
        }
	\and Sahil Singla\thanks{
        (singla@cs.princeton.edu)
        Computer Science Department at Princeton University and School of Mathematics at Institute for Advanced Study. Supported in part by the Schmidt Foundation.
        }
	\and Goran Zuzic\thanks{
        (gzuzic@cs.cmu.edu)
        Computer Science Department,
        Carnegie Mellon University. Supported in part by NSF grants CCF-1527110, CCF-1618280, CCF-1814603, CCF-1910588, NSF CAREER award CCF-1750808, Sloan Research Fellowship and the DFINITY 2018 award.
        }
}

\date{ \today}

\begin{document}
\maketitle

\setlength{\abovedisplayskip}{2pt}
\setlength{\belowdisplayskip}{2pt}

\newcommand{\cent}{\textsf{center}}
\newcommand{\I}{\mathbb{I}}
\newcommand{\OPT}{\ensuremath{\mathrm{OPT}}\xspace}
\newcommand{\ALG}{\ensuremath{\mathrm{ALG}}\xspace}
\newcommand{\checkp}{T}
\newcommand{\at}[1]{^{(#1)}} 
\newcommand{\ocinter}[1]{\ensuremath{\langle#1]}} 
\newcommand{\cointer}[1]{\ensuremath{[#1\rangle}} 
\newcommand{\oointer}[1]{\ensuremath{\langle#1\rangle}} 
\newcommand{\defeq}{:=}

\newcommand{\calA}{\mathcal{A}}
\newcommand{\calB}{\mathcal{B}}
\newcommand{\calH}{\mathcal{H}}
\newcommand{\calT}{\mathcal{T}}
\newcommand{\calE}{\mathcal{E}}
\newcommand{\val}{v}     
\newcommand{\Int}{\ensuremath{[0,1]}}
\newcommand{\bott}{{\small\mathsf{bot}}}

\newcommand{\light}{{\small\mathsf{light}}}


\begin{abstract}
\medskip
  In classical secretary problems, a sequence of $n$ elements arrive 
  in a uniformly random order, and we want to choose a single item, or a
  set of size $K$. 
  The random order model allows us to escape from the strong lower
  bounds for the adversarial order setting, and excellent algorithms are
  known in this setting. However, one worrying aspect of these
  results is that the algorithms overfit to the model: they are not very
  robust. Indeed, if a few ``outlier'' arrivals are adversarially placed
  in the arrival sequence, the algorithms perform poorly. E.g.,
  Dynkin's popular $1/e$-secretary algorithm is sensitive to even a
  single adversarial arrival: if the adversary gives one large bid at
  the beginning of the stream, the algorithm does not select any element
  at all.

  We investigate a robust version of the secretary problem.
  In the \emph{Byzantine Secretary} model, we have two kinds of
  elements: green (good) and red (rogue). The values of all elements are
  chosen by the adversary. The green elements arrive at times uniformly
  randomly drawn from $[0,1]$. The red elements, however, arrive at
  adversarially chosen times. Naturally, the algorithm does not see
  these colors: how well can it solve secretary problems?

  We show that selecting the highest value red set, or the single
  largest green element is not possible with even a small fraction of
  red items. However, on the positive side, we show that these are the
  only bad cases, by giving algorithms which get value comparable to the
  value of the optimal green set \emph{minus the largest green
    item}. (This benchmark reminds us of regret minimization and digital
  auctions, where we subtract an additive term depending on the ``scale'' of the problem.)
  Specifically, we give an algorithm to pick $K$ elements that gets
  within $(1-\e)$ factor of the above benchmark, as long as
  $K \geq \poly(\e^{-1} \log n)$. We extend this to the knapsack secretary
  problem, for large knapsack size $K$.
  
  For the single-item case, an analogous benchmark is the
  value of the second-largest green item.
  For value-maximization, we give a
  $\poly \log^* n$-competitive algorithm, using a multi-layered
  bucketing scheme that adaptively refines our estimates of second-max
  over time.  For probability-maximization, we show
  the \emph{existence} of a good randomized algorithm, using the minimax
  principle.

  We hope that this work will spur further research on robust algorithms for
  the secretary problem, and for other problems in sequential
  decision-making, where the existing algorithms are not robust and often tend to overfit
  to the model.
\end{abstract}

\newpage





\section{Introduction}
\label{sec:intro}

In sequential decision-making, we have to serve a sequence of requests
\emph{online}, i.e., we must serve each request before seeing the next
one. E.g., in online auctions and advertising, given a sequence of
arriving buyers, we want to choose a high bidder. Equivalently, given
a sequence of $n$ numbers, we want to choose the highest of these. The
worst-case bounds for this problem are bleak: choosing a random buyer
is the best we can do. So we make (hopefully reasonable) stochastic
assumptions about the input stream, and give algorithms that work well
under those assumptions.

A popular assumption is that the values/bids are chosen by an
adversary, but presented to the algorithm in a uniformly random
order. This gives the \emph{secretary} or the \emph{random-order
  model},  
under which we can get much better results.
E.g., Dynkin's secretary algorithm that selects the first prefix-maximum
bidder after discarding the first $1/e$-fraction of the
bids, selects the highest bid with probability
$1/e$~\cite{Dynkin-Journal63}. The underlying idea---of fixing one or
more 
thresholds after seeing some prefix of the elements---can be 
generalized to solve classes of packing linear programs
near-optimally~\cite{DevanurHayes09,Devanur11,KRTV14,GM-MOR16}, and to
get $O(\log\log n)$-competitive algorithms for
matroids~\cite{Lachish-FOCS14,FSZ-SODA15} in the random-order model.

However, the assumption that we see the elements in a uniformly random
order is quite strong, and most current algorithms are not robust to
small perturbations to the model. 
 E.g., Dynkin's
algorithm is sensitive to even a single adversarial corruption: if the
adversary gives one large bid at the beginning of the stream, the
algorithm does not select any buyer at all, even if the rest of the
stream is perfectly
random! 
Many other algorithms in the secretary model suffer from similar
deficiencies, which suggests that we may be over-fitting to the
assumptions of the model. 


We propose the \emph{Byzantine secretary model}, where the goal is to
design algorithms robust to outliers and adversarial changes. The use
of the term ``Byzantine'' parallels its use in distributed systems,
where some of the input is well-behaved while the rest is arbitrarily
corrupted by an adversary. Alternatively, our model can be called
\emph{semi-random} or \emph{robust}: these other terms are used in the
literature which inspires our work. Indeed, there is much interest
currently in designing stochastic algorithms that are robust to
adversarial noise (see~\cite{Dia-tut,Moitra18,DiakonikolasKK016,LaiRV16,CharikarSV17,Moitra18,DiakonikolasKK018,EKM-TEAC18,LykourisML18}
and references therein).  Our work seeks to extend robustness to
online problems. Our work is also related in spirit to investigations into how much randomness in the
stream is necessary and sufficient to get competitive
algorithms~\cite{ChungMV13, KKN-STOC15}.

\subsection{Our Model}
\label{sec:model-results}

In the secretary problem, $n$ elements arrive one-by-one.  Each item
has a value that is revealed upon its arrival, which happens at a time chosen
independently and \emph{uniformly at random} in $[0,1]$. (We choose the
continuous time model, instead of the \emph{uniformly random} arrival
order model, since the independence allows us to get clean proofs.) 
When we see an item, we must either \emph{select} it or discard it before we
see the next item. Our decisions are \emph{irrevocable}. We can select
at most $K$ elements, where $K=1$ for the classical version of the
problem. We typically want to maximize the expected total value of the
selected elements where the value of a set is simply the sum of values
of individual elements. (For the single-item case we may also want to
maximize the probability of selecting the highest-value item, which is
called the \emph{ordinal case}.)
Given its generality and wide applicability, this model and its
extensions are widely studied; see \S\ref{sec:related}.





The difference between the classical and Byzantine secretary models
is in how the sequence is generated. In both models, the adversary
chooses the values of all $n$ elements. In the classical model, these
are then permuted in a random order (say by choosing the arrival times
independently and uniformly at random (u.a.r.) from $\Int$). In the Byzantine model, the elements are divided
into two groups: the \emph{green (or good)} elements/items $G$, and the
\emph{red (or rogue/bad)} elements/items $R$. This partition and the
colors are not visible to the algorithm. Now elements in $G$ arrive at
independently chosen u.a.r.\ times between $\Int$, but those in $R$
arrive at times chosen by the adversary. Faced with this
sequence, the algorithm must select some subset of elements (say,
having size at most $K$, or more generally belonging to some down-closed
family).

The precise order of actions is important: 
\begin{OneLiners}
\item First, the adversary chooses values of elements in $R \cup
  G$, and the arrival times of elements in $R$.
\item Then each element $e\in G$ is independently assigned a uniformly
  random arrival time $t_e \sim U[0,1]$. 
\end{OneLiners}
Hence the adversary is powerful and strategic, and can ``stuff'' the
sequence with values in an order that fools our algorithms the
most. The green elements are non-strategic (hence are in random order) and
beyond the adversary's control. When an element is presented, the
algorithm does not see the color (green vs.\ red): it just sees the
value and the time of arrival. We assume that the algorithm knows $n
\defeq |R| + |G|$, but not $|R|$ or $|G|$; see Appendix~\ref{sec:relaxn}
on how to relax this assumption.  The green elements are denoted $G =
\{ g_{\max} =  g_1, g_2, \ldots, g_{|G|} \}$ in non-increasing order of values.

What results can we hope to get in this model? Here are two cautionary examples:
\begin{itemize}
\item Since the red elements behave completely arbitrarily, the adversary can give
  non-zero values to only the reds, and plant a bad example for the
  adversarial order using them. Hence, we cannot hope to get the value
  of the optimal red set in general, and should aim to
  get value only from the
  greens. 
\item Moreover, suppose essentially all the value among the greens is
  concentrated in a single item $g_{\max}$. Here's a bad example: the
  adversary gives a sequence of increasing reds, all having value much
  smaller than $g_{\max}$, but values which are very far from each other. When the algorithm does see the green item,
  it will not be able to distinguish it from the next red, and hence
  will fail. This is formalized in \Cref{obsLbd}.
  Hence, to succeed,
  the green value must be spread among more than one item.
\end{itemize}

Given these examples, here is the ``leave-one-out'' benchmark we propose:
\begin{gather}
  \boxed{V^* := \text{value of the best feasible green set from } G \setminus
  g_{\max}.} \label{eq:1}
\end{gather}
This benchmark is at least as strong as the following 
guarantee:
\begin{gather}
  \text{(value of best feasible green set from $G$) } - v({g_{\max}}). \label{eq:2}
\end{gather}
The advantage of~(\ref{eq:1}) over~(\ref{eq:2}) is that $V^*$ is
interesting even when we want to select a single item, since it asks for
value $v_{g_2}$ or higher.

We draw two parallels to other commonly used benchmarks. Firstly, the
perspective~(\ref{eq:2}) suggests the regret-type guarantees, where we
seek the best solution in hindsight, minus the ``scale of the problem
instance''. The value of $g_{\max}$ is the scale of the instance
here. Secondly, think of the benchmark~(\ref{eq:1}) as assuming the
existence of at least two high bids, then the second-largest element is
almost as good a benchmark as the top element. This is a commonly used
assumption, e.g., in digital goods auctions~\cite{CGL-STOC14}. 

Finally, if we really care about a benchmark that includes $g_{\max}$, 
 our main results for selecting multiple items (Theorem~\ref{thmUnif} and Theorem~\ref{thmKnap}) continue to hold, under the (mild?) assumption
 that the algorithm starts with a polynomial approximation to $v({g_{\max}})$.

\subsection{Our Results}
\label{sec:our-results}


We first consider the setting where we want to select at most $K$
elements to maximize the expected total
value. 
In order to get within $(1+\e)$ factor of the benchmark $V^*$
defined in~(\ref{eq:1}), we need to assume that we have a ``large
budget'', i.e., we are selecting a sufficiently large number of
elements. Indeed, having a larger budget $K$ allows us to make some
mistakes and yet get a good expected value.



\begin{restatable}[Uniform Matroids]{theorem}{thmUnif}
  \label{thmUnif}
  There is an algorithm for Byzantine secretary on uniform matroids of
  rank $K \geq \poly(\e^{-1} \log n)$ that is $(1+\e)$-competitive with the
  benchmark $V^*$.
\end{restatable}

For the standard version of the problem, i.e. without any red elements,
 \cite{Kleinberg-SODA05} gave an
algorithm that achieves  the same competitiveness when
$K \geq \Omega(1/\e^2)$. The algorithm from~\cite{Kleinberg-SODA05} uses
a single threshold, that it updates dynamically; we extend this idea to
having several thresholds/budgets that ``cascade down'' over time; we
sketch the main ideas in \S\ref{sec:our-techniques}.
In fact, we give a more general result---an algorithm for the \emph{knapsack} setting
where each element has a size in $[0,1]$, and the total size of elements
we can select is at most $K$. (The uniform-matroids case corresponds to
all sizes being one.) Whereas the main intuition remain unchanged, 
the presence of non-uniform sizes requires a little more care.
\begin{restatable}[Knapsack]{theorem}{thmKnap}
  \label{thmKnap}
  There is an algorithm for Byzantine secretary on knapsacks with size
  at least $K \geq \poly(\e^{-1} \log n)$ (and elements of at most unit size)
  that is $(1+\e)$-competitive with the benchmark $V^*$.
\end{restatable}

As mentioned earlier,  under mild assumptions 
the guarantee in Theorem~\ref{thmKnap}  can be extended
against  the stronger benchmark that includes $g_{\max}$. 
Formally, assuming the algorithm starts with a  $\poly(m)$-approximation to the
value of $g_{\max}$, we get a $(1+\e)$-competitive
algorithm for $K \geq \poly(\eps^{-1} \log(m n))$ against the stronger benchmark.

\medskip\textbf{Selecting a Single Item.}
What if we want to select a single item, to maximize its expected value?
Note that the benchmark $V^*$ is now the value of $g_2$, the
second-largest green item. Our main result for this setting is the following,
where $\log^* n$ denotes the iterated logarithm:
\begin{restatable}[Value Maximization Single-Item]{theorem}{thmValue}
  \label{thm:single-value}
  There is a randomized algorithm for the value-maximization
  (single-item) Byzantine secretary problem which gets an expected value
  at least $(\log^* n)^{-2} \cdot V^*$.
\end{restatable}

Interestingly, our result is unaffected by the corruption level, and works
even if just two elements $g_{\max}, g_2$ are green, and every other
item is red. This is in contrast to many other robustness models where
the algorithm's performance decays with the fraction of bad
elements~\cite{EKM-TEAC18,CharikarSV17, DiakonikolasKS18a,
  LykourisML18}. 
Moreover, our algorithms do not depend on the fraction of bad
items. Intuitively, we obtain such strong guarantees because the
adversary has no incentive to present too many bad elements with large
values, as otherwise an algorithm that selects a random element would
have a good performance.

In the classical setting, the proofs for the value-maximization proceed
via showing that the best item itself is chosen with constant
probability.  Indeed, in that setting, the competitiveness of 
value-maximization and probability-maximization versions is the same.
We do not know of such a result in the Byzantine model. However, we 
can show a non-trivial performance for the probability-maximization
(ordinal) problem:

\begin{restatable}[Ordinal Single-item Algorithm]{theorem}{thmUbdOrd}
  \label{thmUbdOrd}
  There is a randomized algorithm for the ordinal Byzantine secretary which
  selects an element of value at least the second-largest green item with
  probability $\Omega(1/\log^2 n)$.
\end{restatable}

\medskip\textbf{Other Settings.}
Finally, we consider some other constraint sets given by matroids. In
(simple) partition matroids, the universe $U$ is partitioned into $r$ groups, and
the goal is to select one item from each group to maximize the total
value. If we were to set the benchmark to be the sum of second-largest
green items from each group, we can just run the single-item algorithm
from Theorem~\ref{thmUnif} on each group independently. But our
benchmark $V^*$ is much higher: among the items
$g_2, \ldots, g_{|G|}$, the set $V^*$ selects the largest one from each
group. Hence, we need to get the largest green item from $r-1$ groups!
Still, we do much better than random guessing.

\begin{restatable}[Partition Matroids]{theorem}{thmPartitionMatroid}
  \label{thmPartitionMatroid}
  There is an algorithm for Byzantine secretary on partition matroids that is ${O(\loglog n)^2}$-competitive with the benchmark $V^*$.
\end{restatable}

Finally, we record a simple but useful logarithmic competitive ratio for
arbitrary matroids (proof in \S\ref{sec:general-matroids-proof}),
showing how to extend the corresponding result from~\cite{BIK-SODA07}
for the non-robust case.

\begin{restatable}[General Matroids]{observation}{obsGenMatroid}
  \label{obsGenMatroid}
  There is an algorithm for Byzantine secretary on general matroids that is ${O(\log n)}$-competitive with the benchmark $V^*$.
\end{restatable}


Our results show how to get robust algorithms for the widely-studied
secretary problems, and we hope it will generate futher interest in
robust algorithm design. Interesting next directions include improving
the quantitative bounds in our results (which are almost certainly not
optimal), and understanding tradeoffs between competitiveness and robustness.

\subsection{Related Work}
\label{sec:related}

The secretary problem has a long history, see~\cite{Ferguson-Journal89}
for a discussion.  The papers
\cite{BIK-SODA07,Lachish-FOCS14,FSZ-SODA15} studied generalizations of
the secretary problem to matroids,
\cite{GM-SODA08,KorulaPal-ICALP09,KRTV-ESA13,GS-IPCO17} studied
extensions to matchings, and \cite{Rubinstein-STOC16,RS-SODA17} studied
extensions to arbitrary packing constraints.
More generally, the random-order model has been considered, both as a
tool to speed up algorithms (see~\cite{CS, Seidel}), and to go beyond
the worst-case in online algorithms
(see~\cite{Meyerson-FOCS01,GGLS-SODA08,GHKSV-ICALP14}).  E.g., we can
solve \emph{linear programs} online if the columns of the constraint
matrix arrive in a random
order~\cite{DevanurHayes09,Devanur11,KRTV14,GM-MOR16}, and its entries
are small compared to the bounds.  In online algorithms, the
random-order model provides one way of modeling benign nature, as
opposed to an adversary hand-crafting a worst-case input sequence; this
model is at least as general as i.i.d.\ draws from an unknown distribution.

Both the random-order model and the Byzantine model are
\emph{semi-random}
models~\cite{BlumS-JALG95,FK-JCSS01}, 
with different levels of power to the adversary.
Other restrictions of the random-order model have been studied: the 
model that is closest to ours in spirit is the $t$-bounded adversary model~\cite{guha2009stream},
where the adversary can allowed to delay up to $t$ elements at any time.
This is an adaptive model, where the adversary sees the randomness in the
stream, but is bounded to a small number of changes; we allow the adversary
to change the stream before the elements are randomly placed, but do not
parameterize by the number of changes. The $t$-bounded model has been used
for approximate quantile selection~\cite{guha2009stream}, and for facility location problems~\cite{lang2018online}.
Another line of enquiry lower-bounds the entropy of the input
stream~\cite{ChungMV13,KKN-STOC15}
to ensure the
permutations are ``random enough'';
these papers give sufficient conditions for the classical algorithms to
perform well, whereas we take the dual approach of permitting outliers
and then asking for new robust algorithms.
There are works (e.g.,~\cite{MNS-EC07,MGZ-SODA12,Molinaro-SODA17}) that
give algorithms which have a worst-case adversarial bound, and which
work better when the input is purely stochastic; most of these do not
study the performance on mixed arrival sequences. One exception is the
work~\cite{EKM-TEAC18} who study online matching for mixed arrivals,
under the assumption that the ``magnitude of noise'' is bounded. Another
exception is a recent (unpublished) work of Kesselheim and Molinaro, who
define a robust $K$-secretary problem similar to ours.  They assume the
corruptions have a bursty pattern, and get $1 - f(K)$-competitive algorithms.
Our model is directly inspired by theirs.





\section{Preliminaries and Techniques}
\label{sec:prelims}


By $[a \ldots b]$  we denote the set of integers $\{a, a+1, \ldots, b-1, b\}$.
The \emph{universe} $U \defeq R \cup G$ consists of \emph{red/corrupted} elements $R$ and
\emph{greed/good} elements $G$. Let $v(e)$ denote the value of element $e$: in
the \emph{ordinal} case $v(e)$ merely defines a total ordering on the elements,
whereas in the \emph{value-maximization} case $v(e) \in \mathbb{R}_{\geq
  0}$. Similarly, let $v(\calA)$ be the random variable denoting the value of the
elements selected by algorithm $\calA$. Let $t_e \in \Int$ be the
arrival time of element $e$. Let $R = \{ r_1, r_2, \ldots, r_{|R|} \}$
and $G = \{ g_{\max}, g_2, g_3, \ldots, g_{|G|} \}$; the elements in
each set are ordered in non-increasing values. 
Let $V^*$ be the benchmark to which we compare our algorithm.
Note that $V^*$ is some function of $G \setminus \{ g_{\max} \}$,
depending on the setting. We sometimes refer to elements
$\{e \in U \mid v(e) \geq v(g_2) \}$ as \textit{big}.


\subsection{Two Useful Subroutines}
\label{sec:two-useful}

Here are two useful subroutines. 
\paragraph*{Select a Random Element.}
The subroutine is simple: \emph{select an element uniformly at
  random}. The algorithm can implement this in an online fashion since
it knows the total number of elements $n$. An important property of
this subroutine is that, in the value case, if any element in $U$ has
value at least $n V^*$, this subroutine gets at least $V^*$ in
expectation since this highest value element is selected with
probability $1/n$.
\paragraph*{Two-Checkpoints Secretary.}
The subroutine is defined on two checkpoints
$\checkp_1, \checkp_2 \in \Int$, and let
$\I \defeq [\checkp_1, \checkp_2]$ be the interval between them. The
subroutine ignores the input up to time $\checkp_1$, observes it
during $\I$ by setting threshold $\tau$ to be the highest value seen in the
interval $\I$, i.e., $\tau \gets \max \{ v(e) \mid t_e \in \I
\}$. Finally, during $\ocinter{\checkp_2, 1}$ the subroutine selects
the first element $e$ with value $v(e) \geq \tau$.

We use the subroutine in the single-item setting where the goal is to
find a ``big'' element, i.e.,
an element with value at least 
$v(g_2)$. Suppose that there are no big red elements in $\I$. Now, if
$g_2$ lands in $\I$, and also $g_{\max}$ lands in
$\ocinter{\checkp_2, 1}$, we surely select some element with value at
least $g_2$. Indeed, if there are no big items, threshold $\tau \gets g_2$, and
because $g_{\max}$ lands after $\I$, it or some other element will be
selected. Hence, with probability
$\Pr[t_{g_2} \in \I] \Pr[t_{g_{\max}} \in \ocinter{\checkp_2, 1}] =
(\checkp_2 - \checkp_1) \cdot (1 - \checkp_2)$, we select an element
of value at least $v(g_2)$.

\subsection{Our Techniques}
\label{sec:our-techniques}

\IGNORE{\color{red} 
A common theme of our approaches is to design a \emph{small}
suite of subroutines and show that on any instance one of them has a good performance. The
algorithm selects one of these subroutines at random, hence achieving
almost the same guarantee as the best one (up to the number of
subroutines).
}

A common theme of our approaches is to 
prove a ``good-or-learnable'' lemma for each problem. 
Our algorithms begin by putting
down a small number of \emph{checkpoints}
$\{ \checkp_i \}_i$ to partition the time horizon $[0,1]$---and 
the arriving items---into disjoint \emph{intervals} $\{\I_i\}_i$. We maintain
\emph{thresholds} in each interval to decide whether to select the next element.
Now a ``good-or-learnable'' lemma 
says that either the setting of the thresholds in the current interval  $\I_i$
will give us a ``good'' performance, or we can ``learn'' that this is not the case
and update the thresholds for the next interval $\I_{i+1}$. 
Next we give details for each of our problems.

\medskip\textbf{Uniform Matroid Value Maximization
  (\S\ref{sec:knapsack}).}  Recall that here we want to pick $K$
elements (in particular, all elements have size $1$, unlike the
knapsack case where sizes are in the range $[0, 1]$).  For simplicity,
suppose the algorithm knows that the benchmark $V^*$ lies in $[1,n]$;
we remove this assumption later.  We define $O(\eps^{-1} \log n)$
levels, where level $\ell \geq 0$ corresponds to values in the range
$[n/(1+\e)^{\ell +1}, n/(1+\e)^{\ell})$. For each interval $\I_i$ and
level $\ell$, we maintain a budget $B_{\ell,i}$.  Within this interval
$\I_i$, we select the next arriving element having a value in some
level $\ell$ only if the budget $B_{\ell,i}$ has not been used up. How
should we set these budgets?  If there are $1/\delta$ intervals of
equal size, we expect to select $\delta K$ elements in this
interval. So we have a total of $\delta K$ budget to distribute among
the various levels. We start off optimistically, giving all the budget
to the highest-value level. Now this budget gradually \emph{cascades}
from a level $\ell$ to the next (lower-value) level $\ell+1$, if level
$\ell$ is not selecting elements at a ``good enough'' rate. The
intuition is that for the ``heavy'' levels (i.e., those that contain
many elements from the benchmark-achieving set~$S^*$), we will roughly
see the right number of them arriving in each interval. This allows us
to prove a good-or-learnable lemma, that either we select elements at
a ``good enough'' rate in the current interval, or this is not the
case and we ``learn'' that the budgets should cascade to lower value
levels. There are many details to be handled: e.g., this reasoning is
only possible for levels with many benchmark elements, and so we need
to define a dedicated budget to handle the ``light'' levels.

\IGNORE{\color{red} The goal is now to select $r$ elements, and be competitive with the
leave-one-out benchmark: $V^* \defeq \sum_{i=2}^{r+1} v(g_i)$. As before,
we first read a constant fraction of input in order to estimate $V^*$
via the observed maximum $M$, and infer that $M \le n V^*$. And again,
we can group elements into ``buckets'' and ignore all buckets lying
$O(\log nr)$ levels below $M$.

The new idea here is not to choose one bucket randomly, which loses a
logarithmic factor, but to perform ``adaptive thresholding''. We set the
initial threshold $\tau$ to the lowest bucket. Whenever an element
arrives into some bucket lying above the current threshold, we select it
and \emph{increment the threshold} to the following bucket. 
Furthermore, to hedge against the adversary
tricking us into setting too high a bar, we decrement the threshold after
every $1/r$ units of time, never going below the lowest bucket. This
decrementing is necessary, else algorithm would under-perform even on
inputs where all elements are green and have equal value. We show this
adaptive thresholding algorithm has value $\Omega(V^*)$ when
$r = \Omega(\log n)$. Intuitively, the last condition arises because $\Theta(\log n)$
is the time we need to find the ``correct'' threshold level $\tau$.
}

\medskip\textbf{Single-Item Value-Maximization (\S\ref{sec:card-case}).}
We want to maximize the expected value of the selected element,
compared to $V^* \defeq v(g_2)$, the value of the second-max
green. With some small constant probability our algorithm selects a
uniformly random element. This allows us to assume that every element
has value less than $n V^*$, as otherwise the expected value of a
random guess is $\Omega(V^*)$. We now describe how applying the above
``good-or-learnable'' paradigm in a natural way guarantees an expected
value of $\Omega(V^* /\log n)$. Running the two-checkpoint secretary
(with constant probability) during $\checkp_1 = 0, \checkp_2 = 1/2$ we know that
it gets value $\Omega(V^*)$ and we are done, or failing that, there
exist a red element of value at least $V^*$ in $[0, 1/2]$. But then we
can use this red element (highest value in the first half) to get a
factor $n$ estimate on the value of $V^*$.
\IGNORE{In more detail, suppose we have a single checkpoint $\checkp_1$ at
$t=1/2$: we show an $O(\log n)$ approximation. By the arguments above,
the maximum value $M$ observed before $\checkp_1$ gives us a reasonable
estimate on $V^*$: i.e., $M \in \cointer{V^*, n V^*}$.}
So by  grouping elements into buckets if their values are within a factor 2, 
and randomly guessing the bucket that contains $v(g_2)$,
gives us an expected value of $\Omega(V^* /\log n)$.
To obtain the stronger factor of $\poly\log^* n$ in
Theorem~\ref{thm:single-value}, we now define $\log^* n$ checkpoints.
We  prove a ``good-or-learnable'' lemma that either  selecting a random
element from one of the current buckets has a good value, or we can learn 
a tighter estimate on $V^*$ and  reduce the 
number of buckets.

\medskip\textbf{Ordinal Single-Item Secretary (\S\ref{sec:ordinal_polylog}).}
We now want to maximize the probability of selecting an element whose
value is as large as the green second-max; this is  more challenging
than value-maximization since there is no notion of values for bucketing. Our approach
is crucially different. Indeed, we use the \emph{minimax principle} in
the ``hard'' direction: we give an algorithm that does well when the input
distribution  is \emph{known to the algorithm} (i.e., where the
algorithm can adapt to the distribution), and hence infer the existence
of a (randomized) algorithm that does well on worst-case inputs.


The known-distribution algorithm uses $O(\log n)$ intervals. Again, 
 we can guarantee there is a ``big''
(larger than $g_2$) red element within each interval, as
otherwise running Dynkin's  algorithm on a random interval with a 
small probability 
 already gives a ``good'' approximation. This implies that
 even if the algorithm ``learns'' a good estimate of the 
second-max just before the last interval, it will win. This is because
the algorithm can set this estimate of second-max as a threshold,
and it wins by selecting the 
big red element of the last interval. Finally, to learn a good estimate on
the second-max, we again prove a ``good-or-learnable'' lemma.
Its proof  crucially relies on the algorithm knowing
the arrival distribution, since that allows us to set ``median'' of the 
conditional distribution as a threshold.

\IGNORE{ To find the second-max, it
keeps a set of possible candidate elements for the second-max and
calculates the probability $p$ of their being the second-max (from the
known distribution). It then finds the center $c$ of the candidates
and randomly either (i) selects the first value larger than $c$, or
(ii) goes on to the next interval. The failure of subroutine (i)
forces the input distribution to send most of its probability either
above or below $c$, halving the number of candidates in subroutine
(ii). Differentiating between the ``below'' and ``above'' cases
requires us to calculate the mass $p$, which needs knowledge of the
distribution. After $O(\log n)$ such intervals, the number of
candidates drops from $O(n)$ down to $O(1)$. Designing a constructive
algorithm that achieves similar guarantees remains an interesting open
question.
}

\IGNORE{(In future work we would like to avoid this loss, and
hence have to combine the best of these algorithms, perhaps using ideas
like in~\cite{BlumB00}.)}

\medskip\textbf{Other Results (\S\ref{sec:multItem}).} 
We also give $O(\loglog n)^2$-competitive algorithms for Partition matroids, 
where the difficulty is that we cannot afford to lose the max-green element in
every part. Our idea is to only lose one element globally to
 get a very rough scale of the problem, and then exploit this scale in every part.
We also show why 
other potential benchmarks are too optimistic in \S\ref{sec:lower_bounds}, and
how to relax the assumption that $n$ is known in \S\ref{sec:relaxn}. 
See those sections for details.



\section{Knapsack Byzantine Secretary}
\label{sec:knapsack}


Consider a knapsack of size $K$; for all the results in this section we
assume that $K \geq \poly(\eps^{-1} \log n)$. Each arriving element $e$ has a size
$s(e)\in [0,1]$ and a value $v(e) \geq
0$. 
Let $g_{\max}, g_2, g_3, \ldots, $ denote the green elements $G$ with
decreasing values and let
\begin{gather}
  \textstyle V^* := \max \big\{{\sum_{e\in S} v(e) \mid S \subseteq G
    \setminus g_{\max}\text{ and } \sum_{e\in S} s(e)} \leq K \big\} \label{eq:3}
\end{gather}
be the value of the benchmark solution, i.e., the optimal solution
obtained after discarding the top green element $g_{\max}$. Let $S^*$ be
the set of green elements corresponding to this benchmark.

In \S\ref{subsec:mulItemsWithAssum} we give a $(1+\e)$-competitive algorithm
assuming we have a factor $\poly(n)$-approximation to the benchmark
value $V^*$. (In fact, given this $\poly(n)$-approximation, we can even
get within a $(1+\e)$-factor of the optimal set \emph{including}
$g_{\max}$.) Then in \S\ref{subsec:mulItemsRemAssum} we remove the
assumption, but now our value is only comparable to $V^*$ (i.e., 
excluding $g_{\max}$).


\medskip\textbf{Intuition.}
The main idea of the regular (non-robust)  multiple-secretary
problem (where we pick at most $K$ items) is to observe a small $\e$ fraction
of the input, estimate the value of the $K^{th}$ largest element, and
then select elements with value exceeding this estimate. (A better
algorithm revises these estimates over time, but let us ignore this
optimization for now.) In the Byzantine case, there may be an
arbitrary number of red items, so strategies that try to estimate some
statistics (like the $K^{th}$ largest) to use for the rest of the algorithm are 
susceptible to adversarial attacks.

For now, suppose we know that all items of $S^*$ have values in
$[1, n^c]$ for some constant $c$. The density of an item to be
its value divided by its size. We define $O(\log n)$ \emph{density levels}, where elements in the same level
have roughly the same density, so our algorithm does not distinguish
between them.  The main idea of our algorithm is to use
\emph{cascading budgets}.  At the beginning we allocate all our budget to picking
only the highest-density level items. If we find
that we are not picking items at a rate that is ``good enough'', we
re-allocate parts of our
budget to lower-density levels. 
The hope is that if the benchmark solution $S^*$
selects many elements from a certain density level, we can get a good
estimate of the ``right'' rate at which to pick up items from this
level. Moreover, since our budgets trickle from higher to lower
densities, the only way the adversary can confuse us is by giving
dense red elements, in which case we will select them.

Such an idea works only for the value levels that contain many
elements of $S^*$.  For the remaining value levels, we allocate
\emph{dedicated budgets} whose sole purpose is to pick a ``few''
elements from that level, irrespective of whether they are from
$S^*$. By making the total number of levels logarithmic, we argue that
the total amount of dedicated budgets is only $o(K)$, so it does not
affect the analysis for the cascading budget.

\subsection{An Algorithm Assuming a Polynomial Approximation}
\label{subsec:mulItemsWithAssum}
Suppose we know the benchmark $V^*$ to within a polynomial factor: by
rescaling, assume that $V^*$ lies in the range $[1 \ldots n^c]$ for some
constant $c$. This allows us to simplify the instance structure as
follows: Firstly, we can pick all elements of size at most $1/n$, since
the total space usage is at most $n\cdot 1/n = 1$ (recall,  
$K \geq \poly(\e^{-1} \log n)$). Next, we can ignore
all elements with value less than $1/n^2$ because their total value is
at most $1/n \ll 1 \le V^*$. If the \emph{density} of an element is
defined to be the ratio $v(e)/s(e)$, then all remaining elements have
density between $n^{c+1}$ and $n^{-2}$. The main result of this section
is the following:

\begin{lemma} \label{lem:assumingPolyApprox} If $V^*$ lies between $1$
  and $n^c$ for some constant $c$, each element has size at least
  $1/n$ and value at least $1/n^2$, and $K \geq \poly(\e^{-1} \log n)$,
  then there exists a   $(1+O(\e))$-competitive algorithm.
\end{lemma}

The idea of our algorithm is to partition the input into $1/\delta$
disjoint pieces ($\delta$ is a small parameter that will be chosen
later) and try to solve $1/\delta$ ``similar-looking''
instances of the knapsack problem, each with a knapsack of size
$\delta K$. 

\paragraph{The Algorithm.} Define 
\emph{checkpoints} $T_i := \delta i$ and corresponding intervals 
$\I_i := \ocinter{ \checkp_{i-1}, \checkp_i}$ for all
$i \in [1\ldots 1/\delta]$. Define $L:= (1 + \frac{c+3}{\e}\log n)$ \emph{density levels}
as follows: for each integer $\ell \in [0 \ldots L)$, 
density value $\rho_\ell := n^{c+1}/ (1+\e)^{\ell}$. Now density
level $\ell$ corresponds to all densities lying in the range
$(\rho_{\ell +1}, \rho_\ell]$. Note that densities decrease as $\ell$
increases. We later show that the setting of parameters
$K \geq \Omega\big(\frac{L^2 \log \nicefrac{L}{\e}}{\e^4}\big)$ and
$1 / \delta = \Omega(L / \eps)$ suffices.

We maintain two kinds of
\emph{budgets}:
\begin{itemize}
\item \emph{Cascading budgets}: We maintain a budget $B_{\ell,i}$ for each
  density level $\ell$ and each interval~$\I_i$.  For the first interval
  $\I_1$, define $B_{0,1}:=\delta K$, and $B_{\ell,1} := 0$ for
  $\ell>0$. For the subsequent intervals, we will set $B_{\ell,i}$ in an online
  way as described later.
\item \emph{Dedicated budgets}: We maintain a dedicated budget
  $\widetilde{B}_\ell := H$ for each density level
  $\ell$;  we will later show that setting $H := \Omega\big(\frac{L \log
    \nicefrac{L}{\e}}{\e^3}\big)$ suffices.
\end{itemize}


\begin{figure}[htb]
\begin{center}
\includegraphics[width=8cm]{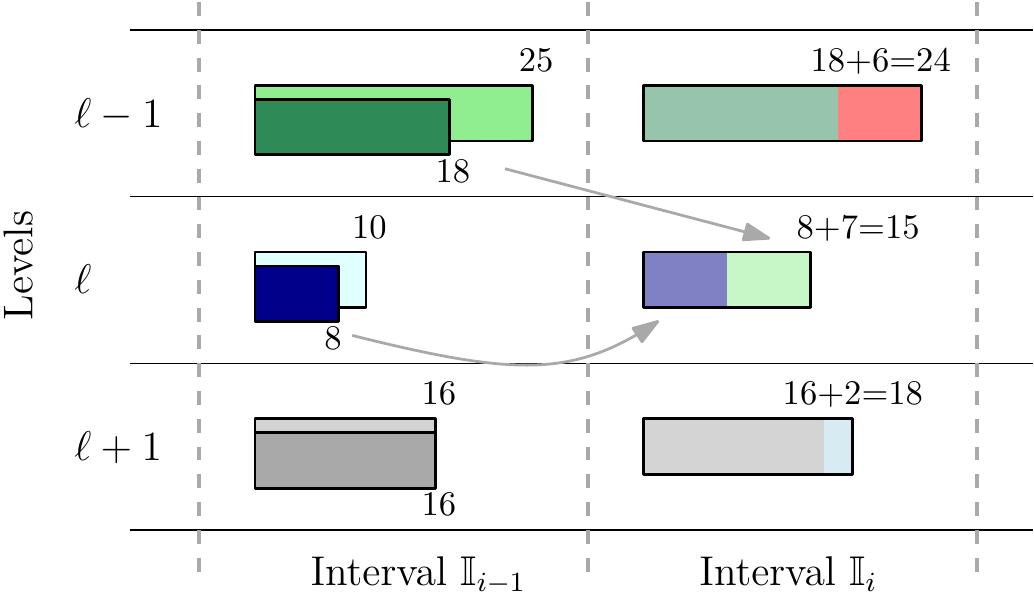}
    \caption{The bars for $\I_{i-1}$ show the budget, and (in darker colors) the amount consumed. The consumed budget (in dark blue) at level $\ell$ in interval $\I_{i-1}$ is
  restored at level $\ell$ in  $\I_{i}$; the unconsumed budget at
  level $\ell-1$ in interval $\I_{i-1}$ is then added to it.}
\label{fig:budgetMovement}
\end{center}
\end{figure}

Suppose we are in the interval $\I_i$, and the arriving element $e$ has density $v(e)/s(e)$ in level $\ell$.
\begin{enumerate}[topsep=0cm,itemsep=0cm]
\item If the remaining cascading budget $B_{\ell',i}$ of one of the density levels
  $\ell' \geq \ell$ is positive
  then select $e$. For the smallest $\ell' \geq \ell$ satisfying this
  condition, update $B_{\ell',i} \gets B_{\ell',i} - s(e)$.

\item Else, if the remaining dedicated budget $\widetilde{B}_{\ell}$ for level $\ell$
  is positive,
  select $e$ and update
  $\widetilde{B}_{\ell} \gets \widetilde{B}_{\ell} - s(e)$.
\end{enumerate}

Finally, for $i > 1$, we define the cascading budgets $B_{\ell,i}$ for
this interval $\I_i$ based on how much of the budgets at levels $\ell$
and $\ell-1$ are consumed in the previous interval $\I_{i-1}$ as follows. The amount of budget $B_{\ell -1,i-1}$ at level $\ell -1$ that is
not consumed in interval $\I_{i-1}$ is moved to level $\ell$ (which has lower
density), and the
budget that gets consumed in $\I_{i-1}$ is restored at level $\ell$ (see \Cref{fig:budgetMovement}).  
Formally, if $C_{\ell, i-1}$  is the amount of \emph{consumed} cascading budget for
level $\ell$  in  interval $\I_{i-1}$ and $R_{\ell,i-1}$ is the amount of \emph{remaining} budget at
level $\ell$ at the end of interval $i-1$ (i.e., the value of $B_{\ell,i-1}$ at
the  time corresponding
to the end of $\I_{i-1}$),  then we define the
initial budget for level $\ell$ at the start of interval $i$ to be
\[ B_{\ell,i} := C_{\ell,i-1} + R_{\ell-1, i-1}. \] 
It is easy to see that we can compute these cascading budgets online.

\emph{A Note about Budget Violations.}  The total budget, summed over
both categories and over all the intervals for the cascading budgets,
is $K' := ((1/\delta)\cdot \delta K) + L\, H > K$. If we use up all
this budget, we would violate the knapsack capacity. Moreover, we
select an element as long as the corresponding budget is positive,
and hence may exceed each budget by the size of the last
element. However, since $K' \le (1 + \eps)K$ and $K$ is much larger
than individual element sizes, the violations is a small fraction of
$K$, so we can \emph{reject} each element originally selected by the
algorithm with some small probability (e.g., $p = 2 \eps$) to
guarantee that the non-rejected selected elements have size at most
$K$ with high probability (i.e., at least $1 - 1/n^c$, for an
arbitrary constant $c > 0$). Henceforth, we will not worry about
violating the budget.

\paragraph{The Analysis.} Recall the benchmark $V^*$ from~(\ref{eq:3}),
and let $S^*$ be a set that achieves this value. All the elements have value
in $[1/n^2,n^c]$ and size at least $[1/n,1]$, so each element $e\in S^*$
has a corresponding density level $\ell(e) \in [0 \ldots L)$ based on its
density $v(e)/s(e)$. We need the notion of ``heavy'' and ``light''
levels. For any level $\ell\in [0 \ldots L)$, define $s^*_\ell$ to be the total
size of elements in $S^*$ with density level $\ell$: 
\begin{gather}
  \textstyle s^*_\ell := \sum_{e \in S^*: \ell(e) = \ell} s(e) . \label{eq:4}
\end{gather}
We say a level $\ell$ is \emph{heavy} if $s^*_\ell \geq H$, else level
$\ell$ is \emph{light}. We refer to (green) elements of $S^*$ at a
heavy (resp., light) level as heavy-green (resp., light-green)
elements. Note that elements not in $S^*$ (some are red and others
green) are left unclassified.
If $H$ is sufficiently large, a concentration-of-measure argument
using the uniformly random arrival times for green
items shows that 
each heavy level receives $(1 - \e)\delta H$ size
during each interval with high probability.  The idea of the proof is to argue that the
cascading budget never ``skips'' a heavy level, and hence we get
almost all the
value of the heavy levels.

To avoid double-counting the values of the
light levels, we separately account for the
algorithm's value attained (a) on light levels using the dedicated
budget or on light-green elements using the cascading budget, and 
(b) for elements that are not light-green (incl.\ red elements) using the
cascading budget. Note that (a) and (b) are disjoint, hence their
contributions can be added up. We show that (a) exceeds the value of
$S^*$ restricted to the light levels, while (b) exceeds $(1 - \eps)$ times the value
of of $S^*$ on the heavy levels. This is sufficient to prove our result. We start by
arguing the former claim.

\begin{claim}[Light-Green Elements]
  \label{claim:lightMass}
  The sum of values of elements selected using the dedicated budget at
  light levels, and of light-green elements selected using the cascading budget,
  is at least
  $ \sum_{\ell \text{ light}} \, s^*_{\ell} \cdot \rho_{\ell+1}$.
\end{claim}
\begin{proof}
  Our algorithm attempts to select each light-green element in $S^*$ 
   using the cascading budget, or failing that, by the dedicated budget at its density level. 
  The only case in which a light-green element $e \in S^*$ is
  dropped is if all the dedicated budget at its level $\ell(e)$ has been
  exhausted. But this means the algorithm has already collected at least
  $s^*_{\ell} \cdot \rho_{\ell+1}$ from the dedicated budget at this
  light level $\ell$.
\end{proof}

Next, to prove that (b) exceeds the value
on heavy levels (up to $1 - \eps$), we need the following property of the 
 cascading budget on the heavy levels. 


\begin{claim} \label{claim:NotFallMuch} For all intervals $\I_i$ and
  levels $\ell$, w.h.p.\ we have that if $B_{\ell,i} > 0$ then
  every heavy level $\ell' < \ell$ satisfies
  $B_{\ell',i}\geq \delta s^*_{\ell'} \cdot (1-\e)$.
\end{claim}
\begin{proof}
  For a heavy level $\ell'$, the expected size of heavy-green elements from
  $S^*$ falling in any interval is $\delta s^*_{\ell'} \ge \delta H$. If
  $\delta H \geq \frac{\Omega(\log (L/(\delta\e)))}{\e^2}$ then with probability $1-\e$
  we get that for each interval $i$ and each heavy level $\ell'$, the total size
  of elements from $S^*$ lying at level $\ell'$ and arriving in interval
  $\I_i$ is at least $\delta s^*_{\ell'} \cdot (1-\e)$, by a
  concentration bound. Henceforth, let us condition on this event
  happening for all heavy levels $\ell'$. 
  
  Now if the cascading budget $B_{\ell,i}>0$, this budget must
  have gradually come from levels  $\ell' < \ell$ of higher densities. But this means
  $B_{\ell',i}\geq \delta s^*_{\ell'} \cdot (1-\e)$ because otherwise
  the cascading budget would never move to level $\ell'+1$, since level $\ell'$
  receives at least $\delta s^*_{\ell'} \cdot (1-\e)$ size of elements
  in every interval.
\end{proof}

For a level $\tau$ let $h^*_{\cointer{0, \tau}} := \sum_{\ell' \text{heavy},~\ell' < \tau } s^*_{\ell'}$ denote the total size of items in $S^*$ restricted to heavy levels from $\cointer{0, \tau}$. Similarly, let $h^{\calA}_{\cointer{0, \tau}}$  be the total size of non-light-green items collected by the algorithm in levels $\cointer{0, \tau}$ and charged against the cascading budget.

\begin{claim}\label{claim:thresholdCascading}
  For all levels $\tau$ we have that $h^{\calA}_{\cointer{0, \tau}} \ge (1 - O(\eps)) h^*_{\cointer{0, \tau}}$.
\end{claim}
\begin{proof}
  Let $t$ be the smallest index of an interval where
  $B_{\tau, t+1} > 0$. We partition the intervals into two groups:
  $\I_1, \I_2, \ldots, \I_{t}$ and $\I_{t+1}, \ldots,
  \I_{1/\delta}$. From \Cref{claim:NotFallMuch} we can conclude that
  for \emph{each interval} in the latter group, the algorithm collects
  a total size of at least
  $\sum_{\ell\text{ heavy},~ \ell < \tau} (1 - \eps) \, \delta s^*_{\ell}
  = (1 - \eps) \delta h^*_{\cointer{0, \tau}}$ from levels
  $\cointer{0, \tau}$. Hence the total contribution over all
  the intervals of the latter group is
  $(\frac{1}{\delta} - t) (1 - \eps)\, \delta\, h^*_{\cointer{0, \tau}} =
  (1 - t \delta)(1 - \eps)\, h^*_{\cointer{0, \tau}}$.

  We now consider the group $\I_1, \ldots, \I_{t}$. Let $C_i, R_i$ and
  $Q_i$ be the total size of the consumed non-light-green, remaining
  budget, and consumed light-green elements charged to the cascading
  budget in interval $i$ with levels $\cointer{0, \tau}$. By
  definition, the total size of all light-green elements is at most
  $LH$, giving $\sum_{i=1}^{t} Q_i \le L H$. Furthermore, since the
  full cascading budget is contained in $\cointer{0, \tau}$, the
  algorithm construction guarantees $C_i + R_i + Q_i = \delta
  K$. Finally, we argue that $\sum_{i=1}^t R_i \le \delta K L$:
  consider an infinitesimally small part $dB$ of the budget. At the
  end of each interval, $dB$ is either used to consume an element or
  it ``moves'' from level $\ell$ to $\ell + 1$, which can happen at
  most $L$ times. Since the total amount of budget per interval is
  $\int dB = \delta K$, the total sum is at most $\delta K L$.

  This lower-bounds the total size contribution of the group $\I_{t+1}, \ldots, \I_{1/\delta}$.
  \begin{align*}
    \sum_{i=1}^t C_i & \quad=\quad t \delta K - \sum_{i=1}^t R_i - \sum_{i=1}^t Q_i \quad \ge \quad t \delta K - \delta K L - L H  \quad \ge \quad (t \delta - \eps - \eps) K \\
               & \quad\ge\quad (t\delta - O(\eps)) \, h^*_{\cointer{0, \tau}},
  \end{align*}
  where we use  $K \ge h^*_{\cointer{0, \tau}}$ (since the total size of elements in $S^*$ is at most $K$), $\delta L \le \eps$, and $LH \le \eps K$.
  Combining contributions from both groups we get:
  \begin{align*}
    (1 - t \delta) (1 - \eps)\, h^*_{\cointer{0, \tau}}  + (t\delta - O(\eps))\, h^*_{\cointer{0, \tau}} & = \left[ (1 - \eps) - t \delta (1 - \eps) + t \delta - O(\eps) \right]\, h^*_{\cointer{0, \tau}} \\
               & = \left( 1 - O(\eps) \right)\, h^*_{\cointer{0, \tau}} .
  \end{align*}
  Hence, we conclude that $h^{\calA}_{\cointer{0, \tau}} \ge (1 - O(\eps))\, h^*_{\cointer{0, \tau}}$.
\end{proof}

Using the above claims we now prove \Cref{lem:assumingPolyApprox}.
\begin{proof}[Proof of Lemma~\ref{lem:assumingPolyApprox}]
  Our fine-grained discretization of densities gives us that
  \begin{gather}
    V^* \le (1 + \eps)\bigg( \sum_{\ell\ \text{light}} s^*_{\ell} \rho_{\ell+1} + \sum_{\ell\ \text{heavy}} s^*_{\ell} \rho_{\ell+1} \bigg) .
  \end{gather}
  From \Cref{claim:lightMass}, our algorithm accrues value at least
  $\sum_{\ell \text{ light}} s^*_{\ell} \cdot \rho_{\ell+1}$ due to
  the elements from light levels that were charged to the dedicated
  budget and light-green elements charged to the cascading budget. It
  is therefore sufficient to prove a similar bound on the value accrued on non-light-green elements charged to the cascading budget with respect to $\sum_{\ell \text{ heavy}} s^*_{\ell} \cdot \rho_{\ell+1}$, which we deduce from \Cref{claim:thresholdCascading}.

  Let $\ell'(x)$ be defined as the largest level $\ell'$ where $\rho_{\ell'} \ge x$, then
  \begin{align*}
    \sum_{\ell\ \text{heavy}} s^*_{\ell} \rho_{\ell+1} \quad = \quad  \int_0^{\infty} \sum_{\ell\text{ heavy, } \rho_{\ell+1} \ge x} s^*_{\ell} \ dx &\quad =\quad  \int_0^{\infty} h^*_{\cointer{0, \ell'(x)}} \ dx \\
   &\quad  \le \quad  (1 + O(\eps)) \int_0^{\infty} h^{\calA}_{\cointer{0, \ell'(x)}} \ dx,
  \end{align*}
  where the last inequality uses \Cref{claim:thresholdCascading}.
  Notice the right-hand side is the value of non-light-green elements charged against the cascading budget. Thus, this part of the algorithm's value exceeds (up to $1 - O(\e)$) the value of heavy levels of $S^*$, finalizing our proof.
\end{proof}

\subsection{An Algorithm for the General Case}
\label{subsec:mulItemsRemAssum}

To 
remove the
assumption that we know a polynomial approximation to $V^*$, the idea is
to ignore the first $\e$ fraction of the arrivals, and use the maximum
value in
this interval to get a $\poly(n)$ approximation to the benchmark. This
strategy is easily seen to work if there are $\Omega(1/\e)$ elements with
a non-negligible contribution to $V^*$. For the other case where most
of the value in $V^*$ comes from a small number of elements, we 
separately reserve some of the budget, and run a collection of
algorithms to catch these elements when they arrive.

Formally,  we define $(1/\e)$ 
\emph{checkpoints} 
$T_i := i\e$ 
and corresponding intervals 
$ \I_i := \ocinter{ \checkp_{i-1}, \checkp_i}$ for all 
$i \in [1\ldots (1/\e)]$. We run the following three algorithms in
parallel, and select the union of elements selected by them.

\begin{enumerate}[label=(\roman*),noitemsep]
\item Select one of the $n$ elements uniformly at random; i.e.,  run Select-Random-Element from \S\ref{sec:two-useful}. \label{alg:knapsackFirst}

\item Ignore elements that arrive in times $[0,\e)$, and let $\hat{v}$
  denote the highest of their values. Run the algorithm from
  \S\ref{subsec:mulItemsWithAssum} during time $[1/\e,1]$, assuming that 
  $V^* \in [\hat{v}/n^2, \hat{v} \, n^2]$. \label{alg:knapsackSecond}

\item At every checkpoint $\checkp_i$, consider the largest value
  $\hat{v}_i$ seen until then. Define $L:= \frac{10}{\e}\log n$
  \emph{value levels} as follows: for $\ell \in (-L/2 \ldots L/2)$ and
  $\tau_\ell :=\hat{v}_i/ (1+\e)^{\ell}$, define level $\ell(i)$ as
  corresponding to values in $(\tau_{\ell(i) +1}, \tau_{\ell(i)}]$. For
  each of these levels $\ell$, keep $D := \frac{10}{\e} \log \frac{1}{\e}$ dedicated
  slots, and select any element having this value level and arriving
  after $\checkp_i$, as long as there is an empty slot in its level. 
  \label{alg:knapsackThird}
\end{enumerate}
The total space used by the there algorithms is at most
\[ 1 + K + (1/\e) \cdot L \cdot D \quad = \quad K + O\Big(\frac{\log n
  \log \nicefrac1\e}{\e^3} \Big)  \quad \leq \quad (1+\e) K,
\]
where the last inequality holds because
$K \geq \Omega\big(\frac{L^2 \log (L/\e)}{\e^4}\big)$ from the size
  condition from \S\ref{subsec:mulItemsWithAssum}. We can now
fit this into our knapsack of size $K$ w.h.p. by sub-sampling each
selected element with probability $(1-O(\e))$. 
To complete the proof of \Cref{thmKnap}, we need to show that we get
expected value $(1-O(\e)) V^*$.

\begin{proof}[Proof of \Cref{thmKnap}]
  The proof considers several cases. Firstly, if there is any single
  element with value more than $n\cdot V^*$, then the algorithm in
  Step~\ref{alg:knapsackFirst} will select it with probability $1/n$,
  proving the claim. Hence, all elements have value at most $n V^*$.

  Now suppose at least $D=\frac{10}{\e} \log \frac{1}{\e}$ elements in
    $S^*$ (recall $S^*$ has    total value
  $V^*$) have individual values at least $V^*/n^2$. In this case, at
  least one of these $D$ elements arrives in the interval $[0,\e)$ with
  probability $1-\e$, and that element gives us the desired
  $n^2$-approximation to $V^*$. Moreover, the expected value of elements
  in $S^*$ arriving in times $[\e,1]$ is at least $(1-O(\e))\,V^*$, even
  conditioning on one of them arriving in $[0,\e)$.

  Finally, consider the case where $D' \leq D$ elements of $S^*$ have
  value more than $V^*/n^2$. The idea of the algorithm in
  Step~\ref{alg:knapsackThird} is to use the earliest arriving of these
  $D'$ elements, or the element $g_{\max}$, to get a rough estimate of
  $V^*$, and from thereon use the dedicated slots to select the remaining
  elements.  Indeed, if  the first of these elements arrive in interval
  $\I_i$, the threshold $\hat{v}_i$ lies in $[V^*/n^2, nV^*]$ (since we
  did not satisfy the first case above). Now the value levels and
  dedicated budgets set up at the end of this interval would pick the
  rest of these $D'$ elements---except those that fall in this same
  interval $\I_i$. We argue that each of the remaining $D'$ elements has
  at least $(1-\eps)$ probability  of not being in $\I_i$, 
  which gives us an expected value of
  $(1-O(\e)) V^*$ in this case as well. This  is true
  because the expected number of these $1+D'$ elements 
  (including $g_{\max}$) that land in any interval that contains
  at least one of them 
   is at most $1+\eps D'$ (even after we condition on the 
  first arrival, each remaining element has  $\eps$ chance
  of falling in this interval). Since any such interval has the same
  chance of being the first interval $\I_i$, and these  $1+D'$ elements
  have the same distribution, the expected number of additional 
  elements  in  $\I_i$ is $\eps D'$.
 %
  %
\end{proof}

This completes the proof of \Cref{thmKnap} for the knapsack case, where
the size $K$ of the knapsack is large enough compared to the largest
size of any element. This generalizes the multiple-secretary problem,
where all items have
unit size. We have not optimized the
value of $K$ that suffices, opting for modularity and simplicity. It can
certainly be improved further, though getting an algorithm that works
under the assumption that $K \geq O(1/\e^2)$, like in the non-robust case, may require new ideas. 



\newcommand{\ts}{\textstyle}

\section{Single-Item Ordinal Case}
\label{sec:ordinal_polylog}

In this section we give a proof of Theorem~\ref{thmUbdOrd}, showing
that there exists an algorithm which selects an element with value no
smaller than $g_2$, with probability at least $\Omega(1/\log^2
n)$. Our proof for this theorem is non-constructive and uses (the hard
direction of) the Minimax Theorem; hence we can currently only show
the \emph{existence} of this algorithm, and not give a compact
description for it. Our main technical lemma furnishes an algorithm
which, given a known (general) probability distribution $\mathcal{B}$
over input instances, selects a big element with probability at least
$\Omega(1/\log^2n)$. Consequently, we use the Minimax Lemma to deduce
that the known-distribution case is equivalent to the worst-case input
setting and recover the analogous result.

Since our algorithms crucially argue about the input distribution
$\mathcal{B}$ and rely on the Minimax, we need to formally define
these terms and establish notation connecting the Byzantine secretary
problem with two-player zero-sum games. Suppose we want to maximize
the probability of selecting a big element and to this end we choose
an algorithm $\calA$, while the adversary chooses a distribution $\calB$ over
the input instances and there is an (infinite) payoff matrix $K$
prescribing the outcomes. Its rows are indexed by different
algorithms, and columns by input instances. Formally, a \emph{``pure''
  input instance} is represented as an $|R|$-tuple of numbers in
$\Int$, representing the arrival times $t_e$ of the red elements; and
a permutation $\pi \in S_n$ over $U$ representing the total ordering
of all values in $U = R \cup G$. Recall that the green elements $G$
choose their arrival times independently and uniformly at random in
$\Int$, hence their $t_e$'s are not part of the input. A
\emph{``mixed'' input instance} is a probability distribution
$\mathcal{B}$ over pure instances $\Int^{|R|} \times S_n$.

While we do not need the full formal specifications of algorithms, we
will mention that a ``mixed'' algorithm $\calA$ is a distribution over
deterministic algorithms. An algorithm $\calA$ on an input instance
$I$ gets a payoff of $K(\calA, I) \defeq \Pr[v(\calA) \ge v(g_2) \mid I]$ where the
probability is taken over the assignment of random arrival times to
elements in $G$ and the distribution of deterministic algorithms
$\calA$. The following Lemma states that for each $\mathcal{B}$ there is an algorithm (that depends on $\mathcal{B}$) that selects a big elements with probability $\Omega(1/\log^2 n)$. We prove the result in~\S\ref{sec:ordinal-algo} and~\S\ref{sec:anal-ordinal}.

\begin{restatable}[Known Distribution Ordinal Single-Item Algorithm]{lemma}{lemDual}
  \label{lemDual}
  Given a distribution over input instances $\mathcal{B}$, there exists an
  algorithm $\calA$ that has an expected payoff of $\Omega(1/\log^2 n)$.
\end{restatable}

To deduce the general case from the known distribution setting, we use a minimax lemma for two-player games. We postpone the details to \Cref{sec:semi-finite-minimax} and simply state the final result here.

\thmUbdOrd*

\subsection{The Algorithm when $\mathcal{B}$ is Known}\label{sec:ordinal-algo}

In this section we give the algorithm for \Cref{lemDual}. We
start with some preliminary notation. For each element $e$, let $t_e$
denote the time at which it appears. Furthermore, for $t \in \Int$,
let $\mathcal{K}(t)$ denote the information seen by the algorithm up
to and including time $t$, consisting of arrival times and relative
values of elements appearing before $t$.

We define $\log n + 1$ time \emph{checkpoints} as follows: set the
initial checkpoint $\checkp_0 \defeq \frac{1}{4}$, and then subsequent
checkpoints $\checkp_i \defeq \frac{1}{4} + \frac{i}{2 \cdot \log n}$
for all $i \in [1\ldots\log n]$. Note that the last checkpoint is
$\checkp_{\log n} = \frac{3}{4}$.  Now the corresponding intervals are
\begin{gather}
  \I_0 \defeq [0, 1/4 ]\quad, \quad \I_i \defeq \ocinter{ \checkp_{i-1},
  \checkp_i} \text{ $\forall \, i \in [1\ldots\log n]$, \quad and } \quad \I_{\log n + 1} \defeq
  \ocinter{ 3/4, 1}.
\end{gather}

Let
$m_i \defeq \max \{ \val(e) \mid e \in R \text { and } t_e \in \I_i
\}$ be the maximum value among the red elements that land in interval
$\I_i$, and let
$\calH \defeq \{ m_i > \val(g_2) \text{ for all
} i \in [1\ldots\log n] \}$ be the event where the maximum value red item in
all intervals is larger than the target $g_2$, i.e., is ``big''. We
call this event $\calH$ the \emph{hard cases} and $\calH^c$ the
\emph{easy cases}; we will show the Two Checkpoints Secretary (from
\S\ref{sec:two-useful}) achieves $\Omega(1/\log^2 n)$ winning probability for
all input instances in $\calH^c$. 
Finally, define
\[ p_e^i \defeq \Pr_{\mathcal{B}}[e = g_2 \mid \calH \text { and }
  \mathcal{K}(\checkp_i)], \] i.e., $p_e^i$ is the probability that
$e$ is the second-highest green element conditioned on the information
seen until checkpoint $\checkp_i$ and the current instance being
hard. Importantly, the algorithm can compute $p^i_e$ at $\checkp_i$.

Now to solve the hard cases, at each checkpoint $\checkp_i$ the
algorithm computes sets $S_i$ satisfying $S_{i+1} \subseteq S_i$. These sets represent elements which are candidates for the second-max. In other words, at time $T_i$ there is reasonable probability that second-max is in $S_i$.
We start with defining $S_0 \gets \{ e \in U \mid t_e \in \I_0 \}$, the
elements the algorithm saw before $\checkp_0$. For $i \ge 0$, let
$c_i$ denote the \emph{center} of $S_i$, i.e., the element of $S_i$
such that there are exactly $\lfloor|S_i|/2\rfloor$ elements smaller
than it. Define $p^i(X) \defeq \sum_{e \in X} p^i_e$ for a set $X$
and index $i \geq 0$.
Given $S_{i-1}$, we determine $S_i$ as follows:
\begin{itemize}
\item Define $\bott_{i-1} \gets \{ e \in S_{i-1}
    \mid \val(e) \leq \val(c_{i-1}) \}$, and note that $\bott_{i-1}
    \subseteq S_{i-1}$.
\item If $p^i(\bott_{i-1}) = p^{i-1}(S_{i-1})$ then $S_i \gets \bott_{i-1}$,
  else $S_i \gets S_{i-1} \setminus \bott_{i-1}$.
\end{itemize}

Our algorithm runs one of the following three algorithms uniformly at random:
\setlist{nolistsep}
\begin{enumerate}[label=(\roman*),noitemsep]
\item Select a random $i \in [1 \ldots \log n]$, define $\tau \gets \max \{
  \val(e) \mid t_e \in \I_i
  \}$ and select the first element larger than $\tau$.
  I.e., run Two Checkpoints Secretary (from \S\ref{sec:two-useful}) with the
  checkpoints being the ends of interval $\I_i$.
\item Select a random $i \in [0 \ldots \log
  n]$, read input until checkpoint $\checkp_i$, define $\tau \gets
  \val(c_i)$ and select the first element larger than it.
\item Compute the sets $S_i$ until $|S_k| \leq 10$ for some
  $k$: then define $\tau$ to be the value of a random element in
  $S_k$, and select the first element larger than it.
\end{enumerate}

\subsection{The Analysis}
\label{sec:anal-ordinal}

In this section we prove \Cref{lemDual}.
Let us give some intuition. We can assume we have a hard case, else the
first algorithm achieves $\Omega(1 / \log^2n)$ winning probability. For the
other two algorithms, let us condition on $g_2$ falling in the first
interval $\I_0$, and then exploit the fact that there is a big red
element in every interval $\I_i$. It may be useful to imagine that we
are trying to guess, at each checkpoint, which of the elements in the
past were actually $g_2$. If we could do this, we would set a threshold
at its value, and select the first subsequent element bigger than the
threshold --- and since there is a $\nicefrac14$ chance that $g_{\max}$
would fall in $\I_{\log n + 1}$, we'd succeed! Of course, since there
are red elements all around, guessing $g_2$ is not straightforward.

So suppose we are at checkpoint $T_k$, and suppose there is a reasonable
probability that $v(g_2) \leq v(c_{k-1})$, but also still some nonzero
probability that $v(g_2) > v(c_{k-1})$. In such a scenario, we claim
that trying to choose an element in the interval $\I_k$ larger than
$c_{k-1}$ will give us a reasonable probability of success. Indeed, we
claim there would have been at least one red element in $\I_k$
bigger than $c_{k-1}$ (since there is still a non-zero probability that
$v(g_2) > v(c_{k-1})$ even at the end of the interval $\I_k$, and since
the case is hard), and $v(g_2) \leq v(c_{k-1})$ with reasonable
probability.  Of course, we only know this at the end of the interval,
but the algorithm can randomly guess $k$ with $\Omega(1/\log n)$
probability. Finally, if there is no such checkpoint, then in every
interval we reduce the size of set $|S_i|$ by half while suffering a
small loss in $p(S_i)$. In this case, both $|S_{\log n}| = O(1)$ and
$p(S_{\log n}) = \Omega(1)$, so the third algorithm can guess $g_2$
with constant probability and select an element larger than it in the
last interval.


\medskip\textbf{Formal Analysis.}  Let $\ALG$ be $1$ if
$v(\calA) \ge v(g_2)$ and $0$ otherwise, where $\calA$ is the
algorithm from the last section. Suppose we're in an easy case, i.e.,
there is an interval $\I_s$ such that all red elements in this
interval are smaller than $g_2$. Now if the first algorithm is chosen,
suppose it selects the interval $\I_s$, suppose $g_2$ lands in $\I_s$,
and $g_{\max}$ lands in $\I_{\log n + 1}$. Then the algorithm surely
selects an element greater than $g_2$, and it has expected value:
\[ \ts \E[\ALG] \geq \frac{1}{3} \cdot \frac{1}{\log n} \cdot \frac{1}{2 \log n} \cdot \frac{1}{4}
  = \Omega\left(\frac1{\log^2 n}\right). \]

Henceforth we can assume the case is hard, and hence each interval $I_i$
contains a red element bigger than $g_2$. We condition on the event that
$g_2$ appears in $\I_0$, which happens with constant probability. Define
\[ \ts k^* \defeq \min \big\{ i \in [1 \dots \log n] \mid \frac{1}{\log n} \leq
  \frac{p^i(\bott_{i-1})}{p^{i-1}(S_{i-1})} < 1 \big\}, \]
and set $k^* = \log n + 1$ if the above set is empty.

\begin{claim}
  \label{clm:one-over-e}
  For all $i<k^*$, the probability $p^i(S_i) = \Omega(1)$.
\end{claim}
\begin{proof}
  By definition, $p^0(S_0) = 1$. By our definition of the sets $S_i$, we
  know that if $p^i(\bott_{i-1}) = p^{i-1}(S_{i-1})$ then
  $p^i(S_i) = p^{i-1}(S_{i-1})$. Else since $i < k^*$, we have
  \[ \ts p^i(S_i) = p^{i-1}(S_{i-1}) - p^i(\bott_{i-1}) \leq
    p^{i-1}(S_{i-1}) (1 - \frac{1}{\log n}). \]
  Hence, $p^i(S_i) \geq (1 - \frac{1}{\log n})^{\log n} = \Omega(1)$, proving the claim.
\end{proof}

Now there are two cases, depending on the value of $k^*$. Suppose
$k^* \leq \log n$. Condition on the event that the second algorithm is
chosen, that it chooses the $i^{th} = (k^*-1)^{th}$ checkpoint, and that
$v(g_2) \leq v(c_i)$. By our choice of $k^*$, we get that
$v(g_2) \leq \tau = v(c_i)$ with probability at least
$p^i(S_i) \cdot \frac{1}{\log n}$, and by Claim~\ref{clm:one-over-e}
this is $\Omega(\frac{1}{\log n})$. Since the case we are considering is
hard and
$\Pr[v(g_2) > v(c_i) \mid \calH \text { and } \mathcal{K}(\checkp_{i+1}) ] > 0$, \\
there is a red element larger than $v(c_i) = \tau$ appearing in $\I_i$. Thus the
algorithm will always select an element in this interval. The correct
interval is chosen with probability $\frac{1}{\log n}$, so the algorithm's
 value is
\[ \ts \E[\ALG] = \frac13 \cdot \frac{1}{\log n} \cdot
  \Omega\big(\frac{1}{\log n}\big) = 
  \Omega\big(\frac{1}{\log^2 n}\big). \]

The other case is when $k^* = \log n + 1$. By definition $|S_0| \leq n$
and $|S_i| \leq \lceil |S_{i-1}|/2 \rceil.$ Therefore
$|S_{\log n}| \leq 10$. Let us condition on the event that the third
algorithm is chosen, that $g_{\max}$ appears in $\I_{\log n + 1}$, and
that the algorithm guesses $g_2$ correctly. The probability of this
event is at least
\[ \ts \frac{1}{3} \cdot \frac{1}{4} \cdot p^{\log n}(S_{\log n}) \cdot
  \frac{1}{10} = \Omega(1). \] where we use Claim~\ref{clm:one-over-e}
to bound the probability $p^{\log n}(S_{\log n})$. In this event, the
algorithm selects an element larger than $g_2$ and has expected value
$\E[\ALG] = \Omega(1).$

Putting all these cases together, we get that our algorithm selects an
element with value at least $v(g_2)$ with probability at least
$\Omega((\log n)^{-2})$. This finishes the proof of \Cref{lemDual}, and
hence of Theorem~\ref{thmUbdOrd}. It remains an intriguing open question
to get a direct algorithm that achieves similar guarantees.



\section{Single-Item Value-Maximization}
\label{sec:card-case}

In this section, we give an algorithm for the problem of selecting an item
to maximize the expected \emph{value}, instead of maximizing the probability of
selecting the second-largest green item (the  ordinal problem considered in
\S\ref{sec:ordinal_polylog}). In the classical secretary problem, both 
problems are well known to be equivalent, with Dynkin's algorithm giving a
tight $1/e$ bound for both. But in the Byzantine case the
problems thus far appear to have different levels of complexity: in
\S\ref{sec:general-matroids-proof} we present a simple $O(\log n)$-competitive
algorithm for the value-maximization byzantine secretary problem, which is already
better than the $\poly\log n$-competitive of \S\ref{sec:ordinal_polylog}. We now
substantially improve it to give a $\poly\log^{*} n$-competitive ratio.

\thmValue*

In the rest of this section, let $V^* \defeq v(g_2)$ denote the
benchmark, the value of the second-largest green element.  The high
level idea of our algorithm is to partition the input into
$O(\log^{*} n)$ intervals and argue that every interval contains a red
element of value $v_i > V^*$, as otherwise Dynkin's algorithm will be
successful.  Moreover, this $v_i$ cannot be much larger than $V^*$, as
otherwise we can just select a random element. This implies we can use
the largest value in each interval to find a good estimate of $V^*$,
and eventually set it as a threshold in the last interval to select a
large value element.

\subsection{The Algorithm} 
\label{algo:value-maximization}
Define $\log\at{i}n$ to be the \emph{iterated logarithm} function:
$\log\at{0} n = n$ and $\log\at{i+1}n = \log(\log\at{i} n)$.
We define $\log^{*} n + 1$ time \emph{checkpoints} as follows: the initial
checkpoint $\checkp_0 = \frac{1}{2}$, and then subsequent checkpoints
$\checkp_i = \frac{1}{2} + \frac{i}{4 \cdot \log^{*} n}$ for all
$i \in [1, \ldots, \log^{*} n]$. Note that the last checkpoint is $\checkp_{\log^{*} n} = \frac{3}{4}$. 
Now the  intervals are
\begin{gather}
  \I_0 \defeq [0, \checkp_0]\quad, \quad \I_i \defeq \ocinter{ \checkp_{i-1},
  \checkp_i} \text{ $\forall \, i \in [1\ldots\log^{*} n]$, \quad and } \quad \I_{\log^{*} n + 1} \defeq
  \ocinter{ \checkp_{\log^{*} n}, 1}.
\end{gather}

Our algorithm runs one of the following three algorithms chosen
uniformly at random.

\setlist{nolistsep}
\begin{enumerate}[label=(\roman*),noitemsep]
\item Select one of the $n$ elements uniformly at random; i.e.,
  run Select-Random-Element from \S\ref{sec:two-useful}.
\label{alg:value-maximizationSecond}

\item Select a random interval $i \in [1 \ldots \log^{*} n]$ and run
  Dynkin's secretary algorithm on $\I_i$. Formally, run Two-Checkpoints-Secretary (from
  \S\ref{sec:two-useful}) with the interval being
  $[\checkp_{i-1}, \frac{1}{2}(\checkp_{i-1} + \checkp_{i})]$.
\label{alg:value-maximizationFirst}

\item Select a random index $i \in [0 \ldots \log^{*} n]$ and observe the
  maximum value during the interval $\I_i$; 
  let this maximum value be $v_{i}$. Choose a uniformly
  random $s \in [0 \ldots 2 \log\at{i} n ]$. Select the first
  element arriving after $\checkp_i$ that has value at least $\tau \defeq
  (v_{i} \log\at{i} n)/2^s$.
\label{alg:value-maximizationThird}
\end{enumerate}

\subsection{The Analysis}
To prove Theorem~\ref{thm:single-value}, assume WLOG that there are only two green elements $g_{\max}$ and $g_2$, and every other element is red (otherwise, we can condition on the arrival times of  all other green elements). Let $v_i$ be the value of the  highest red element in $\I_i$, i.e., excluding $g_{\max}$ and $g_2$.

\begin{proof}[Proof of Theorem~\ref{thm:single-value}]
  We assume $\log\at{i} n$ is an integer for all $i$; this is true with
  a constant factor loss. For sake of a contradiction, assume that the
  algorithm in \S\ref{algo:value-maximization} does not get expected value
  $\Omega((\log^{*} n)^{-2} V^*)$.  Under
  this assumption, we first show that every interval contains a red
  element of  value at least $V^*$.
 
 \begin{claim}\label{claim:cardBigElems}
   For all $j \in [1 \ldots \log^{*} n]$ we have $v_j \ge V^*$. 
  \end{claim}
  \begin{proof}
    Suppose this is not the case. Let $\calE_1$ be the event that the
    following three things happen simultaneously: that we select
    Algorithm~\ref{alg:value-maximizationFirst} in \S\ref{algo:value-maximization} with
    random variable $i = j$, 
    that the second-highest green element $g_2$ falls in the interval
    $\cointer{\checkp_{i-1}, \frac{1}{2}(\checkp_{i-1} + \checkp_{i})}$, and
    that the highest green element $g_{\max}$ falls in
    $ \I_{\log^{*} n + 1}$. Note
    that $\Pr[\calE_1] =  \frac{1}{3\log^{*} n} \cdot \frac{1}{4 \log^{*} n} \cdot \frac14 =
    \Omega((\log^{*} n)^{-2})$. Conditioned on this event $\calE_1$, our algorithm
    (or specifically, Algorithm~\ref{alg:value-maximizationFirst} on the interval
    $\I_{j}$) gets a value at least $v(g_2) = V^*$. Hence the algorithm has
    expected valuation $\Omega\big( (\log^{*} n)^{-2} V^* \big)$, which is a contradiction 
    to our assumption on  its performance.
  \end{proof}

  We now prove that these  red elements with large values cannot be 
  much larger than $V^*$.

  \begin{lemma} \label{lem:cardNotTooBigElems} 
     For all $j \in [1 \ldots \log^{*} n]$ we have $v_j \le V^* \cdot \log\at{j-1} n$. 
  \end{lemma}
  \begin{proof}
    We prove this lemma by induction. The base case $j = 1$ says
    $v_1\le n V^*$, i.e., the highest observed value in $\I_1 = \cointer{\checkp_0, \checkp_1}$ 
    is at most $n V^*$. Suppose this is not the case---there
    exists a red element $e$ in  $\I_1$ with value at
    least $n V^*$. Let $\calE_1$ be the event that  we select
    Algorithm~\ref{alg:value-maximizationSecond} in \S\ref{algo:value-maximization} (i.e.,
    Select-Random-Element) and that it selects $e$. Since
    $\Pr[\calE_1] = \Omega(\frac{1}{n})$, we have a contradiction that the expected valuation 
    is $\Omega(V^*)$.

    Now suppose the statement is true until  $j \ge 1$. We prove the
    inductive step  $j+1$. Suppose not, i.e., $v_{j+1} > V^* \log\at{j} n$.
     Let $\calE_2$ be the event
    that we select Algorithm~\ref{alg:value-maximizationThird} in
    \S\ref{algo:value-maximization} with parameter $i = j$ and that the random $s\in  [0 \ldots 2 \log\at{j} n]$ 
    is  such     that     $ \frac{v_{j} }{2^{s+1}} \leq V^* < \frac{v_{j} }{2^{s}}$ (it exists by induction hypothesis).
    This implies threshold $\tau := \frac{v_{j} \log\at{j} n}{2^s} $  is between $\frac12 V^* \log\at{j} n$ and $V^* \log\at{j} n$. Note 
     $\Pr[\calE_2] \geq \frac13 \cdot \frac{1}{\log^{*} n} \cdot \frac{1}{2 \log\at{j} n}$. Since
    event $\calE_2$ implies the algorithm gets value at least $\tau \geq \frac12 V^* \log\at{j} n$ (because $v_{j+1}>\tau$), its expected value is $\Omega( (\log^{*} n)^{-1} V^* )$, a contradiction.
    \end{proof}

Now by Claim~\ref{claim:cardBigElems} and
Lemma~\ref{lem:cardNotTooBigElems}, we have $v_{j} \in [V^*, V^* \cdot
\log\at{j} n]$ for all $j \in [1 \ldots \log^{*} n]$. We still get a contradiction.
 Let $\calE_3$ be the event that the following three things happen
simultaneously: that we select Algorithm~\ref{alg:value-maximizationThird} in
\S\ref{algo:value-maximization} with $i = \log^{*} n$, that the highest green
element $g_{\max}$ is in interval $\I_{\log^{*} n+1}$, and that we
select $s$ in Algorithm~\ref{alg:value-maximizationThird} such that 
$\tau := ({v_{\log^{*} n} \log\at{j} n})/({2^s}) $ is  between $\frac12 V^*$ and $V^*$. Note 
$\Pr[\calE_3] \geq  \frac{1}{3 \log^{*} n} \cdot \frac{1}{4} \cdot \frac{1}{2 \log^{*} n}$.
Since   the  event $\calE_2$ implies the algorithm gets value at least $\tau \geq \frac12 V^*$ (because $g_{\max}$ is in $\I_{\log^{*} n+1}$), its expected value is $\Omega( (\log^{*} n)^{-2} V^* )$.  Thus, we have a contradiction in every case, which means our assumption is incorrect and the algorithm has expected value  $\Omega((\log^{*} n)^{-2} V^*)$.
\end{proof}



\section{Value Maximization for Matroids}\label{sec:multItem}

In this section we discuss multiple-choice Byzantine secretary
algorithms in the matroid setting. 

\begin{defn}[Byzantine secretary problem on matroids]
  Let $\calM$ be a matroid over $U = R \cup G$, where elements in
  $G = \{ g_{\max}, g_2, \ldots, g_{|G|} \}$ arrive uniformly at random
   in $[0, 1]$. When  an
  element $e \in U$ arrives, the algorithm must irrevocably select or ignore $e$, while
   ensuring that the set of selected elements forms an independent
  set in $\calM$. The leave-one-out benchmark $V^*$ is  the
  highest-value independent subset of $G \setminus \{ g_{\max}
  \}$.
\end{defn}


The knapsack results imply $(1-\e)$-competitiveness for uniform matroids
as long as the rank $r$ is large enough; we now consider other matroids.

\subsection{$O(\loglog n)^2$-competitiveness for Partition Matroids}
\label{sec:part}

A partition matroid is where the elements of the universe are
partitioned into parts $\{P_1,P_2, \ldots\}$. Given some integers
$r_1,r_2, \ldots$, a subset of elements is independent if for every
$i$ it contains at most $r_i$ element from part $P_i$.
\thmPartitionMatroid*
We prove Theorem~\ref{thmPartitionMatroid} for simple partition
matroids where all $r_i = 1$, i.e., we can select at most one element
in each part. This is without loss of generality (up to $O(1)$
approximation) because we can randomly partition each part $P_i$
further into $r_i$ parts and run the simple partition matroid
algorithm.

Recall that our single item $\poly\log^{*} n$ algorithm from \S\ref{sec:card-case} no longer works for partition matroids. This is because besides one part we want to get the highest green element in all the other parts.  Formally, Claim~\ref{claim:cardBigElems} where we use Dynkin's secretary algorithm in the proof of Theorem~\ref{thm:single-value}  fails because it needs  at least two green elements. So we need to overcome the lower bound  to getting  the highest-value green  element $v(g_1)$ in Observation~\ref{obsLbd}. We achieve this and design an $O(\loglog n)^2$-approximation algorithm by making an assumption that the algorithm starts with a polynomial approximation to $v(g_1)$. Although in general this is a strong assumption, it turns out that for partition matroids this assumption is w.l.o.g. because the algorithm may  lose the highest green element in one of the parts.

\subsubsection{The Algorithm}
We define $\loglog n + 1$ time \emph{checkpoints} as follows: the initial
checkpoint $\checkp_0 = \frac{1}{2}$, and then subsequent checkpoints
$\checkp_i = \frac{1}{2} + \frac{i}{2 \cdot \loglog n}$ for all
$i \in [1\ldots \loglog n]$. 
Now the corresponding intervals are
\begin{gather}
  \I_0 \defeq [0, \checkp_0] \quad   \text{and}   \quad \I_i = \ocinter{ \checkp_{i-1},
  \checkp_i} \quad \forall \, i \in [1\ldots \loglog n] 
\end{gather}
Let $v_0$ denote the value of the max element seen by the algorithm in $\I_0$.  

Now for every part $P$ of the partition matroid, we execute the following algorithm separately. Let $v_i$ for $i\in [1\ldots \loglog n]$ denote the value of the max element seen by the algorithm \emph{in part $P$} during interval $\I_i$.  
Let $V^*$ denote the element of our benchmark in $P$. Notice that $v_i \in P$ and $V^*$ cannot be the overall highest green element as we exclude it. 
 We  define $4 \log^{1/i} n$ levels for $\I_i$ where \emph{level} $j$ for $j \in [1\ldots 4\log^{1/i} n]$ is given by elements with values in 
 \[ \big[\frac{v_{i-1} \cdot \log^{1/i} n }{2^{j} } , \frac{v_{i-1} \cdot \log^{1/i} n}{2^{j-1}} \big]. \] 
We run one of the following algorithms uniformly at random.

\begin{enumerate}[label=(\roman*)]
\item Select an  element  uniformly at random as discussed in \S\ref{sec:two-useful}. \label{alg:partOne}
\item For every part $P$, select a random interval $i\in [1\ldots \loglog n]$ and select a random level $j \in [4\log^{1/i} n]$. Select the first element above $\frac{v_{i-1} \cdot \log^{1/i} n }{2^{j} }$ in $P$.  \label{alg:partSecond}
\item For every part $P$, select a random interval $i\in  [1\ldots \log\log n]$ and if there is an element with value more than $2^{\log^{1/i} n}$ times the max of all the already seen elements in $\I_{i}$, selects it with constant probability, say $1/100$.  \label{alg:partThird}
\end{enumerate}

\subsubsection{The Analysis}
Since with constant probability our algorithm selects one of the $n$ elements uniformly at random (Algorithm~\ref{alg:partOne}), we can assume that $ v_0 \leq n^2 \cdot V^*$. We always condition on the event that $g_{\max}$ arrives in the interval $\I_0$, which happens with constant probability and implies $v_0 \geq v(g_{\max})$.
Moreover, we ignore parts $P$ where $V^*$ is below $v(g_{\max})/n^2$ because they do not contribute significantly to the benchmark. So from now assume 
\[ V^*/n^2 \leq  v_0 \leq n^2 \cdot V^*. \] 

We design an algorithm that gets value $\Omega(V^*/(\loglog n)^2)$ in each part $P$, which implies Theorem~\ref{thmPartitionMatroid} by linearity of expectation over parts.

 Let $v_i\at{\mathrm{red}} \in P$ for $i\in [1\ldots \log\log n]$ denote the value of the max red element that the adversary presents in $\I_{i}$. 

\begin{claim}
If there exists an  $i\in [1\ldots \log\log n]$ with $v_i\at{\mathrm{red}} > V^*\cdot \log^{1/i} n$ then the expected value of the algorithm is $\Omega(V^*/\loglog n)$.
\end{claim}
\begin{proof}
With constant probability, our algorithm selects a random interval $i$ and selects a random level element in it (Algorithm~\ref{alg:partSecond}). Since w.p. $1/\loglog n$ it selects this $i$, and w.p.  $\frac14 \log^{1/i} n$ it selects the random level of $v_i\at{\mathrm{red}}$ in $\I_i$, the algorithm has expected value at least 
\[\frac{1}{\loglog n} \cdot \frac{1}{4 \log^{1/i} n} \cdot v_i\at{\mathrm{red}}    \geq    \frac{1}{4\loglog n} \cdot V^*. \qedhere
\] 
\end{proof}

By the last claim we can assume for all $i\in  [1\ldots \log\log n]$, we have $ v_i\at{\mathrm{red}} \leq V^*\cdot \log^{1/i} n$.
\begin{claim} If there exists an  $i\in [1\ldots \log\log n]$  with $v_i\at{\mathrm{red}} < V^*/2^{\log^{1/i} n}$ then the expected value of the algorithm is $\Omega(V^*/(\loglog n)^2)$.
\end{claim}
\begin{proof}
With constant probability the algorithm guesses one of the intervals $i$ and if there is an element with value more than $2^{\log^{1/i} n}$ times the max of all the already seen elements in $\I_{i}$, selects it with constant probability (Algorithm~\ref{alg:partThird}). With $1/\loglog n$ probability the algorithm selects this particular $i$ and with  $1/\loglog n$ probability $V^*$ appears in this interval with value at least $2^{\log^{1/i} n}$ times the max seen element in this interval. Notice there can be at most $O\Big(\frac{4 \log^{1/i} n}{\log^{1/i} n}\Big) = O(1)$ elements with such large jumps in value in this interval. In this case our algorithm selects $V^*$ with constant probability. 
\end{proof}

Finally, we are only left with the case where for all $i\in  [1\ldots \log\log n]$ value 
$\frac{V^*}{2^{\log^{1/i} n}} \leq v_i\at{\mathrm{red}} \leq V^*\cdot \log^{1/i} n,
$
which we handle using Algorithm~\ref{alg:partSecond}.

\begin{claim} If for all $i\in [1\ldots \log\log n ]$ we have 
\[ \frac{V^*}{2^{\log^{1/i} n}} \leq v_i\at{\mathrm{red}} \leq V^*\cdot \log^{1/i} n \]  then the expected value of the algorithm is $\Omega(V^*/(\loglog n)^2)$.
\end{claim}
\begin{proof}
Consider Algorithm~\ref{alg:partSecond}. It selects $i = \loglog n -1$ w.p.  $1/\loglog n$. Moreover, suppose $V^*$ appears in $\I_{\loglog n-1}$. Now since there are only a constant number of levels in this interval, our algorithm  selects an element of value at least $V^*$ with constant probability.
\end{proof}
We have shown that in every case the algorithm has expected value $\Omega(V^*/(\loglog n)^2)$ for any fixed part $P$. This implies Theorem~\ref{thmPartitionMatroid} by linearity of expectation over parts.

\subsection{$O(\log n)$-approx for General Matroids}\label{sec:general-matroids-proof}
\obsGenMatroid*

\begin{proof}
Notice that no element can have weight more than $nr$ times the second max-element because w.p. $1/n$ our algorithm selects one of the $n$ elements uniformly at random. Given this, condition on the event that the max element with value $v$ lands in the first half of the input. Define $2\log (nr)$ exponentially separated levels  as follows: 
\[ \big[ \frac{v}{2^{\log (nr)}}, \frac{v}{2^{\log (nr)-1}} \big\rangle, \big[ \frac{v}{2^{\log (nr)-1}}, \frac{v}{2^{\log (nr)-2}} \big\rangle,  \ldots, \big[ \frac{v}{2}, v \big\rangle, \ldots, \big[ {v}{2^{\log (nr)-1}}, {v}{2^{\log (nr)}} \big\rangle .\]
Since at least one of these intervals contains at least $2\log(nr)$ fraction of $\OPT$, we can guess that interval and run a greedy algorithm, i.e., accept any element with value in that interval or above if it is independent.
\end{proof}



\section{Conclusion}

In this paper we defined a robust model for the secretary problem, one
where some of the elements can arrive at adversarially chosen times,
whereas the others arrive at random times. For this setting, we argue
that a natural is the optimal solution on all but the highest-valued
green item (or even simpler, the optimal solution on the green items,
minus the single highest-value item). This benchmark reflects the fact
that we cannot hope to compete with the red (adversarial) items, and
also cannot do well if all the green value is concentrated in a single
green item. 

We show that for the case where we want to pick $K$ items, or if we have
a knapsack of size $K$, we can get within $(1-\e)$ of this benchmark,
assuming $K$ is large enough. We can also get non-trivial results for
the single-item case, where our benchmark is now the second-highest
valued green item. In the ordinal setting where we only see the relative
order of arriving elements and the goal is to maximize the probability
of getting an element whose value is above the benchmark, we use the
minimax principle to show existence of an $O(\log^2 n)$-approximation
algorithm in \S\ref{sec:ordinal_polylog}. In the value maximization
setting, we give an $O(\log^* n)^2$-approximation algorithm in
\S\ref{sec:card-case}. We also show $O(\log \log n)$-competitiveness for
partition matroids.

The results above suggest many question. Can we improve the lower bound
on the size required for $(1-\e)$-competitiveness? Can we get a
constant-competitive algorithm for the single-item case? For the
probability-maximization problem, our proof only shows the existence of
an algorithm; can we make this constructive? More generally, many of the
algorithms for secretary problems seem to overfit to the model, at least
in the presence of small adversarial changes: how can we make our
algorithms robust?




\subsection*{Acknowledgments}

We thank Thomas Kesselheim and Marco Molinaro for sharing their model
and thoughts on robust secretary problems with us; these have directly inspired our model. 

\medskip
\medskip

\appendix


\section{Hard Benchmarks}
\label{sec:lower_bounds}

We show that for the benchmark $V^* \defeq v(g_{max})$, every algorithm has an approximation of at most $O(1/n)$.

\begin{restatable}[Lower Bound for $g_{\max}$]{observation}{obsLbd}
  \label{obsLbd}
  Any randomized algorithm for the single-item Byzantine secretary
  problem cannot select the highest-value good/green item
  with probability larger than $1/(|R|+1)$.
\end{restatable}


\begin{proof}
  We use Yao's minimax lemma, so it is enough to construct an input
  distribution $\mathcal{B}$ for which no deterministic algorithm can
  achieve an approximation better than $\frac{1}{|R| + 1}$. The
  distribution is as follows. The red elements arrive at random
  times, that is
  $(t_{r_1}, t_{r_2}, ... t_{r_{|R|}}) \sim U[0,1]^{|R|}$. The linear
  ordering among the elements is set such that the red elements are
  strictly increasing according to their arrival time, or formally:
  $t_{r_i} > t_{r_j} \implies v(r_i) > v(r_j)$. The maximum element is
  green and all the other green elements are smaller than all red
  elements. Formally:
  $v(g_{max}) > v(r_1) > v(r_2) ... > v(r_{|R|}) > v(g_2) > v(g_3)
  \ldots > v(g_{|G|})$. This fully defines the input distribution.

  All the arrival times are distinct with probability $1$. Let $\mathcal{K}(t)$ denote the information seen by the algorithm up to and including time $t.$ Partition the probability space according to $S \defeq \{ t_{g_{max}}, t_{r_1}, t_{r_2}, ... t_{r_{|R|}} \}$ and $L \defeq (t_{g_2}, t_{g_3}, ... t_{g_{|G|}}).$ Let $s_1 < s_2 < ... < s_{|R| + 1}$ be the elements of $S.$ Let $M_i \defeq \{t_{g_{max}} = s_i\}.$ By definition, we have $Pr[M_i | S, L] = \frac{1}{|R| + 1}.$ Note that, since the red items arrive in increasing order of value and the green item has maximum value, we have $Pr[M_i | S, L, \mathcal{K}(t)] = Pr[M_j | S, L, \mathcal{K}(t)]$ for all $t \leq s_i, s_j.$ Therefore, $Pr[M_i | S, L, \mathcal{K}(s_i)] \leq \frac{1}{|R|+2-i}.$ In other words, there is no way to distinguish the maximum green element from the red elements before it is too late, that is at the time of the green element's arrival. Thus, by a simple inductive argument, the proof is finished.
\end{proof}

Using techniques presented in~\cite{CDFS-2018}, we can extend this result to the value case as well.



\section{Relaxing the Assumption that $n$ is Known}
\label{sec:relaxn}

In this section we extend our results to some settings where $n$ is unknown. Most importantly, observe that all of the results in this paper hold even if $n$ is \emph{known only up to a constant factor} with at most a constant factor degradation in the quality of the result. As a simple example, note that picking a uniformly random element from an $n$-element sequence when the assumed number of elements is $\tilde{n} \in [n, 2n]$ will select an element $x \in U$ with probability $p_x \in [\frac{1}{2n}, \frac{1}{n}]$, leading to a degradation in the result by a factor of at most $2$, which we typically ignore in this paper.

This still leaves us open to the possibility that we do not even know the scale of $n$. Surprisingly, it is still possible to ``guess'' $\tilde{n}$ while only incurring a loss of $\widetilde{O}(\log n)$ in the quality, even if there is no prior known upper limit on $n$.\footnote{By $\widetilde{O}(f(n))$ we mean $f(n) \cdot \mathrm{poly}(\log f(n))$.} The following claim formalizes this result.

\begin{claim}
  There exists a distribution $X$ over the integers such that for every $n \ge 1$ the probability that the sampled number $\tilde{n} \sim X$ is within a constant factor of $n$, is at least $1 / \widetilde{O}(\log n)$.
\end{claim}
\begin{proof}
  Consider the sequence $a_k \defeq \frac{1}{k (\log_2 k) (\log_2 \log_2 k)^2}$ defined for $k \ge 2$. It is well-known that this sequence converges, i.e., $\sum_{k=2}^\infty a_k = O(1)$. A simple way to see this is by noting that a non-negative decreasing sequence $(A_k)_k$ converges if and only if $(2^k A_{2^k})_k$ converges~\cite[Thm 3.27]{rudin1976principles}.
  
  Let $\sum_k^\infty A_k \sim \sum_k^\infty B_k$ be the equivalence relation denoting that $(A_k)_k$ and $(B_k)_k$ either both converge or both diverge. Then the above fact implies that $\sum_k^\infty \frac{1}{k \log_2 k (\log_2 \log_2 k)^2} \sim \sum_k^\infty \frac{1}{k (\log_2 k)^2} \sim \sum_k^\infty \frac{1}{k^2}$, where the last sequence clearly converges.
  
  We can assume without loss of generality that $n \ge 100$ by handling those cases separately. The strategy for guessing the estimate $\tilde{n}$ is now immediate: we sample $\tilde{n}$ from $\mathbb{Z}_{\ge 2}$ according to the distribution $\Pr[\tilde{n} = k] = a_k / Z$ where $Z \defeq \sum_{k=100}^\infty a_k = O(1)$. We observe that $\Pr[\tilde{n} \in \ocinter{2^{k-1} \ldots 2^k}] \ge \frac{1}{2} 2^k a_{2^k} / Z = \widetilde{\Omega}(1/k)$. Let $k'$ be the unique index such that $n \in \ocinter{2^{k'-1} \ldots 2^{k'}}$, hence $k' = \Theta(\log n)$. Then $\Pr[\tilde{n} \in \ocinter{2^{k'} \ldots 2^{k'+1}}] = \widetilde{\Omega}(1/k') = \widetilde{\Omega}(1/\log n)$. But also in that case we have that $n \in [\tilde{n}/4, \tilde{n}]$ and we are done.
\end{proof}

Finally, consider an important case where the fraction of red elements is bounded away from $1 - \Omega(1)$. This is a reasonable assumption for most applications, e.g., online auctions, where we do not expect that most of the arrivals will be chosen by an adversary. By simply observing the first half of the sequence, i.e., $\cointer{0, 1/2}$, we can typically estimate $n$ up to a constant while degrading the expected output of our algorithms by at most a constant factor.

\begin{claim} If there is a constant $\varepsilon<1$ such that the fraction of red elements $\frac{|R|}{|R| + |G|} \leq \varepsilon$ then we can estimate $n$ up to a constant factor by time $t=1/2$.
\end{claim}
\begin{proof}
  We run a simple preprocessing step to estimate $n$ up to a constant factor by $t=1/2$. Notice that the expected number of green elements to arrive in the interval $\cointer{0,1/2}$ is $0.5 \cdot |G|  = 0.5 \cdot n(1-\varepsilon) = \Omega(n)$. Since by simple Chernoff bounds this means that w.h.p. we see $\Omega(n)$ elements in the first half, we run a simple algorithm that does not select any element till $t=1/2$, and then use the number of elements that arrive in $\cointer{0,1/2}$ as an estimate of $n$. 
\end{proof}



\section{Minimax}\label{sec:semi-finite-minimax}

In this section we argue that an $\alpha$-payoff (i.e., the probability of selecting the second-max element or better is at least $\alpha$) known distribution algorithm for the ordinal single-item Byzantine secretary implies an $\alpha$-payoff algorithm for the general, worst-case input, setting. This can be directly modeled as a two-player game where player A chooses an algorithm $\calA$ and player B chooses a distribution over the input instances $\calB$. Our coveted result would go along the lines of
\begin{align*}
  \sup_{\calA} \inf_{\calB} K(a, b) = \inf_{\calB} \sup_{\calA} K(\calA, \calB),
\end{align*}
where $K(\calA, \calB)$ denotes the payoff when we run algorithm $\calA$ on the input distribution $\calB$. The left-hand side denotes the worst-case input setting, while the right-hand side denotes the known distribution setting.

The main challenge in proving such a claim stems from the infiniteness
of the set of algorithms and set of input distributions. Indeed, if
one makes no finiteness assumption for either $A$ or $B$, the Minimax
property can fail even for relatively well-behaved two-player
games~\cite{parthasarathy1970games}. On the other hand if both $A$ and
$B$ would be finite, then the result would follow from the classic Von
Neumann's Minimax~\cite{neumann1928theorie}.

\begin{fact}[Von Neumann's Minimax]
  \label{fact:von-neumann-minimax}
  Let $A$ and $B$ be finite sets. Denote by $\calD(A)$ and $\calD(B)$ distributions over $A$ and $B$, respectively. Then for any matrix of values $K : A \times B \to \mathbb{R}$ it holds that
  \begin{align}
    \max_{a \in \calD(A)} \min_{b \in \calD(B)} K(a, b) = \min_{b \in \calD(B)} \max_{a \in \calD(A)} K(a, b).
  \end{align}  
\end{fact}

The infiniteness of the sets stems from the arrival times being in the infinite set $\Int$. To solve this issue, we slightly modify our algorithm by discretizing $[0, 1]$. Let $N = n^3, \calT \defeq \{\frac{0}{N}, \frac{1}{N}, \frac{2}{N}, \ldots, \frac{N}{N} \}$ and $\Phi: [0, 1] \rightarrow \calT, \Phi(t) \defeq \lfloor N \cdot t \rfloor / N$ be the discretizing function. We modify our algorithm in the following way: apply $\Phi$ to the input distribution $\mathcal{B}$, as well as to every arrival time. Note that the elements are presented to the algorithm exactly as before, it just $\textit{pretends}$ they arrive in discrete time steps. We can assume $\Phi(t_e) \neq \Phi(t_g)$ for every $e \in U, g \in G, e \neq g$ (otherwise, we say the algorithm loses), since this happens with at most $n^2/N = o(1)$ probability. Using completely analoguous techniques as in \Cref{sec:ordinal_polylog} we can show this algorithm is $\Omega(1 / \log^2 n)$-competitive.

We note that a randomized algorithm is simply a distribution over deterministic algorithms. Hence our second goal is to argue that the number of distinct deterministic algorithms is finite (i.e., bounded by a function of $n$). To this end we have to specify how we represent them with at least some formality. We identify a deterministic algorithm $\calA$ with a function that gets evaluated each time a new element arrives; its parameter is the information history $\mathcal{K}(t)$ ($t$ being the current time) represented in any appropriate format; its output is $\{ \bot, \top \}$ representing whether to select the current element. For concreteness, the information history consists of $(t_e, \pi_e)$ for every element $e$ that arrived before the function call, where $t_e \in \calT$ is the discretized arrival time (after applying $\Phi$) and $\pi_e \in [0 \ldots n-1]$ is the relative value order of $e$ with respect to prior arrived elements. The number of distinct histories is bounded by $((N+1) n)^n$, a function of $n$; therefore the set of deterministic algorithms, i.e., functions from the history to $\{\bot, \top\}$, is also bounded.

We remember that an input distribution is simply a distribution over ``pure'' inputs. Note that the payoff of a deterministic algorithm for a specific input depends only on the following: $\Phi(t_r)$ for every red element; $\pi \in S_n$, the permutation representing the total order among the elements; and $\pi_t \in S_{|R|}$, the permutation denoting the order in which the red elements arrive (since red elements can have the same discretized arrival time, but an arbitrary order in which they are presented to the algorithm). The above discretization makes the number of pure inputs at most $(n!)^2 \cdot (N+1)^n$, i.e., bounded by a function of $n$. The reader can refresh their memory about the representation of pure inputs by reviewing the introduction to \Cref{sec:ordinal_polylog}.

Finally, for our discretized algorithm, we proved that the set of pure inputs with different payoffs, as well as the number of deterministic algorithms is bounded by a function of $n$. Therefore, for a fixed $n$, both numbers are finite. We invoke the Von Neumann's Minimax (Fact \ref{fact:von-neumann-minimax}) to conclude that the best result in the known distribution setting and worst-case input setting are equivalent, recovering the following theorem.

\thmUbdOrd*



{\small
\bibliographystyle{alpha}
\bibliography{../bib}
}

\end{document}